\documentclass[a4paper,11pt]{article}

\usepackage{amssymb}
\usepackage{amsmath}
\usepackage{amsfonts}
\usepackage{amsthm}
\usepackage[pdftex]{hyperref}

\theoremstyle{plain}
\newtheorem{proposition}{Proposition}
\newtheorem{claim}{Claim}
\newtheorem{lemma}{Lemma}

\newtheorem{theorem}{Theorem}
\theoremstyle{definition}
\newtheorem{observation}{Observation}
\newtheorem{definition}{Definition}

\usepackage{fullpage}

\usepackage{xspace}
\usepackage{enumerate}

\newcommand{\dotcup}{\mathbin{\mathaccent\cdot\cup}}
\renewcommand{\P}{\mathcal{P}}
\newcommand{\F}{\mathcal{F}}

\newcommand{\X}{\mathcal{X}}
\renewcommand{\L}{\ensuremath{\mathcal{L}}\xspace}
\newcommand{\N}{\mathbb{N}}

\newcommand{\eqvr}[0]{\ensuremath{\mathcal{R}}\xspace}
\newcommand{\Oh}{\mathcal{O}}

\DeclareMathOperator{\tw}{tw}
\newcommand{\x}{\mathbf{x}}
\newcommand{\Z}{\mathbb{Z}}
\newcommand{\ZZ}{\mathcal{Z}}

\newcommand{\cref}[1]{(\ref{#1})\xspace}

\newcommand{\probname}[1]{\textsc{\lowercase{#1}}}
\newcommand{\problem}[1]{\probname{#1}\xspace}

\newcommand{\ILPF}{\probname{ILPF}\xspace}

\newcommand{\IS}{\probname{Independent Set}\xspace}
\newcommand{\HSN}{\probname{Hitting Set($n$)}\xspace}
\newcommand{\HS}{\probname{Hitting Set}\xspace}
\newcommand{\SAT}{\probname{Satisfiability}\xspace}

\newcommand{\NP}{\ensuremath{\mathsf{NP}}\xspace}
\newcommand{\FPT}{\ensuremath{\mathsf{FPT}}\xspace}
\newcommand{\Ptime}{\ensuremath{\mathsf{P}}\xspace}
\newcommand{\coNP}{\ensuremath{\mathsf{coNP}}\xspace}
\newcommand{\containment}{\ensuremath{\mathsf{NP \subseteq coNP/poly}}\xspace}
\newcommand{\ncontainment}{\ensuremath{\mathsf{NP \nsubseteq coNP/poly}}\xspace}

\newcommand{\introduceproblem}[3]{%
\begin{quote}
#1\\
\textbf{Input:} #2\\
\textbf{Question:} #3
\end{quote}}

\newcommand{\yes}{\emph{yes}\xspace}
\newcommand{\no}{\emph{no}\xspace}

\title{A structural approach to kernels for ILPs:\\ Treewidth and Total Unimodularity}

\author{Bart M.\ P.\ Jansen\thanks{Supported by ERC Starting Grant 306992 while working for the University of Bergen.}\\Eindhoven University of Technology\\The Netherlands\\\texttt{b.m.p.jansen@tue.nl}
\and Stefan Kratsch\thanks{Supported by the German Research Foundation (DFG), KR 4286/1.}\\University of Bonn\\Germany\\\texttt{kratsch@cs.uni-bonn.de}}

\begin{document}

\maketitle

\begin{abstract}
Kernelization is a theoretical formalization of efficient preprocessing for \NP-hard problems. Empirically, preprocessing is highly successful in practice, for example in state-of-the-art ILP-solvers like CPLEX. Motivated by this, previous work studied the existence of kernelizations for ILP related problems, e.g., for testing feasibility of $Ax\leq b$. In contrast to the observed success of CPLEX, however, the results were largely negative. Intuitively, practical instances have far more useful structure than the worst-case instances used to prove these lower bounds.

In the present paper, we study the effect that subsystems that have (a Gaifman graph of) bounded treewidth or that are totally unimodular have on the kernelizability of the ILP feasibility problem. We show that, on the positive side, if these subsystems have a small number of variables on which they interact with the remaining instance, then we can efficiently replace them by smaller subsystems of size polynomial in the domain without changing feasibility. Thus, if large parts of an instance consist of such subsystems, then this yields a substantial size reduction. Complementing this we prove that relaxations to the considered structures, e.g., larger boundaries of the subsystems, allow worst-case lower bounds against kernelization. Thus, these relaxed structures give rise to instance families that cannot be efficiently reduced, by any approach.
\end{abstract}

\section{Introduction}\label{section:intro}

The notion of kernelization from parameterized complexity is a theoretical formalization of preprocessing (i.e., data reduction) for \NP-hard combinatorial problems. Within this framework it is possible to prove worst-case upper and lower bounds for preprocessing; see, e.g., recent surveys on kernelization~\cite{Kratsch14,LokshtanovMS12}.
Arguably one of the most successful examples of preprocessing in practice are the simplification routines within modern integer linear program (ILP) solvers like CPLEX (see also~\cite{AtamturkS05,GenovaG11_survey,Savelsbergh94}). Since ILPs have high expressive power, already the problem of testing feasibility of an ILP is \NP-hard; there are immediate reductions from a variety of well-known \NP-hard problems. Thus, the problem also inherits many lower bounds, in particular, lower bounds against kernelization.

\introduceproblem{\problem{Integer Linear Program Feasibility} -- \ILPF}{A matrix $A\in\Z^{m\times n}$ and a vector~$b\in\Z^m$.}{Is there an integer vector~$x\in \Z^n$ with $Ax\leq b$?}

Despite this negative outlook, a formal theory of preprocessing, such as kernelization aims to be, needs to provide a more detailed view on one of the most successful practical examples of preprocessing, even if worst-case bounds will rarely match empirical results. With this premise we take a structural approach to studying kernelization for \ILPF. We pursue two main structural aspects of ILPs. The first one is the treewidth of the so-called \emph{Gaifman graph} underlying the constraint matrix~$A$.  
As a second aspect we consider ILPs whose constraint matrix has large parts that are totally unimodular.
Both bounded treewidth and total unimodularity of the whole system $Ax\leq b$ imply that feasibility (and optimization) are tractable.\footnote{Small caveat: For bounded treewidth this also requires bounded domain.} We study the effect of having \emph{subsystems} that have bounded treewidth or that are totally unimodular. We determine when such subsystems allow for a substantial reduction in instance size. Our approach differs from previous work~\cite{Kratsch13_ilp2,Kratsch13_ilp1} in that we study structural parameters related to treewidth and total unimodularity rather than considering parameters such as the dimensions $n$ and $m$ of the constrain matrix $A$ or the sparsity thereof.

\paragraph{Treewidth and ILPs.} The Gaifman graph~$G(A)$ of a matrix~$A\in\Z^{m\times n}$ is a graph with one vertex per column of~$A$, i.e., one vertex per variable, such that variables that occur in a common constraint form a clique in~$G(A)$ (see Section~\ref{section:gaifman:decomposition}). 
This perspective allows us to consider the structure of an ILP by graph-theoretical means. In the context of graph problems, a frequently employed preprocessing strategy is to replace a simple (i.e., constant-treewidth) part of the graph that attaches to the remainder through a constant-size boundary, by a smaller gadget that enforces the same restrictions on potential solutions. There are several meta-kernelization theorems (cf.~\cite{KimLPRRSS13}) stating that large classes of graph problems can be effectively preprocessed by repeatedly replacing such \emph{protrusions} by smaller structures. 
It is therefore natural to consider whether large protrusions in the Gaifman graph~$G(A)$, corresponding to subsystems of the ILP, can safely be replaced by smaller subsystems. 

We give an explicit dynamic programming algorithm to determine which assignments to the boundary variables (see Section~\ref{section:treewidth:protrusionreduction}) of the protrusions can be extended to feasible assignments to the remaining variables in the protrusion. Then we show that, given a list of feasible assignments to the boundary of the protrusion, the corresponding subsystem of the ILP can be replaced by new constraints. If there are~$r$ variables in the boundary and their domain is bounded by~$d$, we find a replacement system with~$\Oh(r \cdot d^r)$ variables and constraints that can be described in~$\tilde\Oh(d^{2r})$ bits. 
By an information-theoretic argument we prove that equivalent replacement systems require~$\Omega(d^r)$ bits to encode. 
Moreover, we prove that large-domain structures are indeed an obstruction for effective kernelization by proving that a family of instances with a single variable of large domain (all others have~$\{0,1\}$), and with given Gaifman decompositions into protrusions and a small shared part of encoding size~$N$, admit no kernelization or compression to size polynomial in~$N$. 

On the positive side, we apply the replacement algorithm to protrusion decompositions of the Gaifman graph to shrink \ILPF instances. When an \ILPF instance can be decomposed into a small number of protrusions with small boundary domains, replacing each protrusion by a small equivalent gadget yields an equivalent instance whose overall size is bounded. The recent work of Kim et al.~\cite{KimLPRRSS13} on meta-kernelization has identified a structural graph parameter such that graphs from an appropriately chosen family with parameter value~$k$ can be decomposed into~$\Oh(k)$ protrusions. If the Gaifman graph of an \ILPF instance satisfies these requirements, the \ILPF problem has kernels of size polynomial in~$k$. Concretely, one can show that bounded-domain \ILPF has polynomial kernels when the Gaifman graph excludes a fixed graph~$H$ as a topological minor and the parameter~$k$ is the size of a modulator of the graph to constant treewidth. We do not pursue this application further in the paper, as it follows from our reduction algorithms in a straight-forward manner.

\paragraph{Total unimodularity.} Recall that a matrix is totally unimodular (TU) if every square submatrix has determinant $1$, $-1$, or $0$. If $A$ is TU then feasibility of $Ax\leq b$, for any integral vector $b$, can be tested in polynomial time. (Similarly, one can efficiently optimize any function $c^Tx$ subject to $Ax\leq b$.) We say that a matrix $A$ is \emph{totally unimodular plus $p$ columns} if it can be obtained from a TU matrix by changing entries in at most $p$ columns.
Clearly, changing a single entry may break total unimodularity, but changing only few entries should still give a system of constraints $Ax\leq b$ that is much simpler than the worst-case. Indeed, if, e.g., all variables are binary (domain $\{0,1\}$) then one may check feasibility by simply trying all $2^p$ assignments to variables with modified column in $A$. The system on the remaining variables
will be TU and can be tested efficiently.

From the perspective of kernelization it is interesting whether a small value of $p$ allows a reduction in size for $Ax\leq b$ or, in other words, whether one can efficiently find an equivalent system of size polynomial in $p$. We prove that this depends on the structure of the system on variables with unmodified columns. If this remaining system decomposes into separate subsystems, each of which depends only on a \emph{bounded number of variables in non-TU columns}, then by a similar reduction rule as for the treewidth case we get a reduced instance of size polynomial in $p$ and the domain size $d$. Complementing this we prove that in general, i.e., without this bounded dependence, there is no kernelization to size polynomial in $p+d$; this also holds even if $p$ counts the \emph{number of entry changes} to obtain $A$ from a TU matrix, rather than the (usually smaller) number of modified columns.

\paragraph{Related work.} Several lower bounds for kernelization for \ILPF and other ILP-related problems follow already from lower bounds for other (less general) problems. For example, unless \containment and the polynomial hierarchy collapses\footnote{\ncontainment is a standard assumption in computational complexity. It is stronger than $\Ptime \neq\NP$ and $\NP\nsubseteq\coNP$, and it is known that \containment implies a collapse of the polynomial hierarchy.}, 
there is no efficient algorithm that reduces every instance~$(A,b)$ of \ILPF to an equivalent instance of size polynomial in~$n$ (here~$n$ refers to the number of columns in~$A$); this follows from lower bounds for \HS~\cite{DomLS14} or for \SAT~\cite{DellM14} and, thus, holds already for binary variables ($0/1$-\ILPF). The direct study of kernelization properties of ILPs was initiated in~\cite{Kratsch13_ilp2,Kratsch13_ilp1} and focused on the influence of row- and column-sparsity of~$A$ on having kernelization results in terms of the dimensions~$n$ and~$m$ of~$A$. At high level, the outcome is that unbounded domain variables rule out essentially all nontrivial attempts at polynomial kernelizations. In particular, \ILPF admits no kernelization to size polynomial in~$n+m$ when variable domains are unbounded, unless \containment; this remains true under strict bounds on sparsity~\cite{Kratsch13_ilp2}. For bounded domain variables the situation is a bit more positive: there are generalizations of positive results for $d$-\HS and $d$-\SAT (when sets/clauses have size at most $d$). One can reduce to size polynomial in~$n$ in general~\cite{Kratsch13_ilp1}, and to size polynomial in~$k$ when seeking a feasible~$x\geq 0$ with~$|x|_1\leq k$ for a sparse covering ILP~\cite{Kratsch13_ilp2}.

\paragraph{Organization.}
Section~\ref{section:preliminaries} contains preliminaries about parameterized complexity, graphs, and treewidth. 
In Section~\ref{section:treewidth} we analyze the effect of treewidth on preprocessing ILPs, while we consider the effect of large totally unimodular submatrices in Section~\ref{section:tum}. In Section~\ref{section:tudiscussion} we discuss some differences between totally unimodular and bounded-treewidth subsystems. We conclude in Section~\ref{section:discussion}.

\section{Preliminaries}\label{section:preliminaries}

\paragraph{Parameterized complexity and kernelization.} A \emph{parameterized problem} is a set~$Q\subseteq\Sigma^*\times\N$ where~$\Sigma$ is any finite alphabet and~$\N$ denotes the non-negative integers. In an instance~$(x,k)\in\Sigma^*\times\N$ the second component is called the \emph{parameter}. A parameterized problem~$Q$ is \emph{fixed-parameter tractable} (\FPT) if there is an algorithm that, given any instance~$(x,k)\in\Sigma^*\times\N$, takes time~$f(k)|x|^{\Oh(1)}$ and correctly determines whether~$(x,k)\in Q$; here~$f$ is any computable function. A \emph{kernelization} for~$Q$ is an algorithm~$K$ that, given~$(x,k)\in\Sigma^*\times\N$, takes time polynomial in~$|x|+k$ and returns an instance~$(x',k')\in\Sigma^*\times\N$ such that~$(x,k)\in Q$ if and only if~$(x',k')\in Q$ (i.e., the two instances are equivalent) and~$|x'|+k'\leq h(k)$; here~$h$ is any computable function, and we also call it the \emph{size of the kernel}. If~$h(k)$ is polynomially bounded in~$k$, then~$K$ is a \emph{polynomial kernelization}. We also define \emph{(polynomial) compression}; the only difference with kernelization is that the output is any instance~$x'\in\Sigma'^*$ with respect to any fixed language~$\L$, i.e., we demand that~$(x,k)\in Q$ if and only if~$x'\in\L$ and that~$|x'|\leq h(k)$. A polynomial-parameter transformation from a parameterized problem~$P$ to a parameterized problem~$Q$ is a polynomial-time mapping that transforms each instance~$(x,k)$ of~$P$ into an equivalent instance~$(x',k')$ of~$Q$, with the guarantee that~$(x,k)\in P$ if and only if~$(x',k')\in Q$ and~$k'\leq p(k)$ for some polynomial~$p$.

\paragraph{Lower bounds for kernelization.} For one of our lower bound proofs we use the notion of a \emph{cross-composition} from \cite{BodlaenderJK14}, which builds on the framework for lower bounds for kernelization by Bodlaender et al.~\cite{BodlaenderDFH09} and Fortnow and Santhanam \cite{FortnowS11}.

\begin{definition}\label{definition:polyequivalencerelation}
An equivalence relation~\eqvr on $\Sigma^*$ is called a \emph{polynomial equivalence relation} if the following two conditions hold:
\begin{enumerate}
	\item There is an algorithm that given two strings~$x,y \in \Sigma^*$ decides whether~$x$ and~$y$ belong to the same equivalence class in~$(|x| + |y|)^{\Oh(1)}$ time.
	\item For any finite set~$S \subseteq \Sigma^*$ the equivalence relation~$\eqvr$ partitions the elements of~$S$ into at most~$(\max _{x \in S} |x|)^{\Oh(1)}$ classes.
\end{enumerate}
\end{definition}
\begin{definition}\label{definition:crosscomposition}
Let~$L \subseteq \Sigma^*$ be a set and let~$Q \subseteq \Sigma^* \times \mathbb{N}$ be a parameterized problem. We say that~$L$ \emph{cross-composes} into~$Q$ if there is a polynomial equivalence relation~$\eqvr$ and an algorithm that, given~$t$ strings~$x_1, x_2, \ldots, x_t$ belonging to the same equivalence class of~$\eqvr$, computes an instance~$(x^*,k^*) \in \Sigma^* \times \mathbb{N}$ in time polynomial in~$\sum _{i\in[t]} |x_i|$ such that:
\begin{enumerate}
	\item~$(x^*, k^*) \in Q \Leftrightarrow x_i \in L$ for some~$i \in [t]$,
	\item~$k^*$ is bounded by a polynomial in~$\max _{i \in [t]} |x_i|+\log t$.
\end{enumerate}
\end{definition}
\begin{theorem}[\cite{BodlaenderJK14}] \label{theorem:crosscomp:nokernel}
If the set~$L \subseteq \Sigma^*$ is NP-hard under Karp reductions and~$L$ cross-composes into the parameterized problem~$Q$, then there is no polynomial kernel or compression for~$Q$ unless \containment.
\end{theorem}

\paragraph{Graphs.} All graphs in this work are simple, undirected, and finite. For a finite set~$X$ and positive integer~$n$, we denote by~$\binom{X}{n}$ the family of size-$n$ subsets of~$X$. The set~$\{1,\ldots,n\}$ is abbreviated as~$[n]$. An undirected graph~$G$ consists of a vertex set~$V(G)$ and edge set~$E(G) \subseteq \binom{V(G)}{2}$. For a set~$X \subseteq V(G)$ we use~$G[X]$ to denote the subgraph of~$G$ induced by~$X$. We use~$G - X$ as a shorthand for~$G[V(G) \setminus X]$. For~$v \in V(G)$ we use~$N_G(v)$ to denote the open neighborhood of~$v$. For~$X \subseteq V(G)$ we define~$N_G(X) := \bigcup _{v \in X} N_G(v) \setminus X$. The \emph{boundary} of~$X$ in~$G$, denoted~$\partial_G(X)$, is the set of vertices in~$X$ that have a neighbor in~$V(G) \setminus X$.

\paragraph{Treewidth and protrusion decompositions.}
A \emph{tree decomposition} of a graph~$G$ is a pair~$(T, \X)$, where~$T$ is a tree and~$(\X = \{X_i \mid i \in V(T)\})$ is a family of subsets of~$V(G)$ called \emph{bags}, such that (i) $\bigcup_{i \in V(T)} X_i = V(G)$, (ii) for each edge~$\{u,v\} \in E(G)$ there is a node~$i \in V(T)$ with~$\{u,v\} \subseteq X_i$, and (iii) for each~$v \in V(G)$ the nodes~$\{i \mid v \in X_i\}$ induce a connected subtree of~$T$. The \emph{width} of the tree decomposition is~$\max _{i\in V(T)} |X_i| - 1$. The \emph{treewidth} of a graph~$G$, denoted~$\tw(G)$, is the minimum width over all tree decompositions of~$G$.
An optimal tree decomposition of an $n$-vertex graph~$G$ can be computed in time~$\Oh(2^{\Oh(\tw(G)^3)} n)$ using Bodlaender's algorithm~\cite{Bodlaender96_lintimetw}. A $5$-approximation to treewidth can be computed in time~$\Oh(2^{\Oh(\tw(G))}n)$ using the recent algorithm of Bodlaender et al.~\cite{BodlaenderDDFLP13}. A vertex set~$X$ such that~$\tw(G - X) \leq t$ is called a \emph{treewidth-$t$ modulator}. 

For a positive integer~$r$, an \emph{$r$-protrusion} in a graph~$G$ is a vertex set~$X \subseteq V(G)$ such that~$\tw(G[X]) \leq r - 1$ and~$\partial_G(X) \leq r$. An~$(\alpha, r)$-protrusion decomposition of a graph~$G$ is a partition~$\P = Y_0 \dotcup Y_1 \dotcup \ldots \dotcup Y_\ell$ of~$V(G)$ such that (1) for every~$1 \leq i \leq \ell$ we have~$N_G(Y_i) \subseteq Y_0$, (2) $\max(\ell, |Y_0|) \leq \alpha$, and (3) for every~$1 \leq i \leq \ell$ the set~$Y_i\cup N_G(Y_i)$
is an $r$-protrusion in~$G$. We sometimes refer to~$Y_0$ as the \emph{shared part}.

\section{ILPs of bounded treewidth}\label{section:treewidth}

We analyze the influence of treewidth for preprocessing \ILPF. In Section~\ref{section:gaifman:decomposition} we give formal definitions to capture the treewidth of an ILP, and introduce a special type of tree decompositions to solve ILPs efficiently. In Section~\ref{section:treewidth:basic} we study the parameterized complexity of \ILPF parameterized by treewidth. Tractability turns out to depend on the domain of the variables. An instance~$(A,b)$ of \ILPF has \emph{domain size~$d$} if, for every variable~$x_i$, there are constraints~$-x_i \leq d'$ and~$x_i \leq d''$ for some~$d' \geq 0$ and~$d'' \leq d-1$. (All positive results work also under more relaxed definitions of domain size~$d$, e.g., any choice of~$d$ integers for each variable, at the cost of technical complication.)
The feasibility of bounded-treewidth, bounded-domain ILPs is used in Section~\ref{section:treewidth:protrusionreduction} to formulate a protrusion replacement rule. It allows the number of variables in an ILP of domain size~$d$ that is decomposed by a~$(k,r)$-pro\-tru\-sion decomposition to be reduced to~$\Oh(k \cdot r \cdot d^{r})$.
In Section~\ref{section:tw:limitations} we discuss limitations of the protrusion-replacement approach.

\subsection{Tree decompositions of linear programs} \label{section:gaifman:decomposition}
Given a constraint matrix~$A\in\Z^{m\times n}$ we define the corresponding \emph{Gaifman graph}~$G=G(A)$ as follows~\cite[Chapter 11]{FlumG06}. We let~$V(G) =\{x_1,\ldots,x_n\}$, i.e., the variables in~$Ax \leq b$ for~$b \in \Z^m$. We let~$\{x_i,x_j\}\in E(G)$ if and only if there is an~$r\in [m]$ with~$A[r,i]\neq 0$ and~$A[r,j]\neq 0$. Intuitively, two vertices are adjacent if the corresponding variables occur together in some constraint.

\begin{observation} \label{observation:rows:clique}
For every row~$r$ of~$A \in \Z^{m \times n}$, the variables~$Y_r$ with nonzero coefficients in row~$r$ form a clique in~$G(A)$. Consequently (cf.~\cite{Bodlaender98}), any tree decomposition~$(T,\X)$ of~$G(A)$ has a node~$i$ with~$Y_r \subseteq X_i$.
\end{observation}

To simplify the description of our dynamic programming procedure, we will restrict the form of the tree decompositions that the algorithm is applied to. This is common practice when dealing with graphs of bounded treewidth: one works with \emph{nice tree decompositions} consisting of \emph{leaf, join, forget}, and \emph{introduce} nodes. When using dynamic programming to solve ILPs it will be convenient to have another type of node, the \emph{constraint node}, to connect the structure of the Gaifman graph to the constraints in the ILP. To this end, we define the notion of a \emph{nice Gaifman decomposition} including constraint nodes.

\begin{definition} \label{definition:gaifman:decomposition}
Let~$A \in \Z^{m \times n}$. A \emph{nice Gaifman decomposition of~$A$ of width~$w$} is a triple $(T, \X = \{X_i \mid i \in V(T)\}, \ZZ = \{Z_i \mid i \in V(T)\})$, where~$T$ is a rooted tree and~$(T,\X)$ is a width~$w$ tree decomposition of the Gaifman graph~$G(A)$ with:\\
(1) The tree~$T$ has at most~$4n + m$ nodes.\\
(2) Every row of~$A$ is assigned to exactly one node of~$T$. If row~$j$ is mapped to node~$i$ then~$Z_i$ is a list of pointers to the nonzero coefficients in row~$j$.\\
(3) Every node~$i$ of~$T$ has one of the following types:%
	\begin{description}
		\item[leaf:] $i$ has no children and~$|X_i| = 1$,
		\item[join:] $i$ has exactly two children~$j,j'$ and~$X_i = X_j = X_{j'}$,
		\item[introduce:] $i$ has exactly one child~$j$ and~$X_i = X_j \cup \{v\}$ with~$v \in V(G(A)) \setminus X_j$,
		\item[forget:] $i$ has exactly one child~$j$ and~$X_i = X_j \setminus \{v\}$ with~$v \in V(G(A)) \cap X_j$,
		\item[constraint:] $i$ has exactly one child~$j$,~$X_i = X_j$, and~$Z_i$ stores a constraint of~$A$ involving variables that are all contained in~$X_i$.
	\end{description}
\end{definition}

The following proposition shows how to construct the Gaifman graph~$G(A)$ for a given matrix~$A$. It will be used in later proofs.

\begin{proposition} \label{proposition:build:gaifman:graph}
Given a matrix~$A\in\Z^{m\times n}$ in which each row contains at most~$r$ nonzero entries, the~$n\times n$ adjacency matrix of~$G(A)$ can be constructed in~$\Oh(nm + r^2m + n^2)$ time.
\end{proposition}

\begin{proof}
Initialize an all-zero~$n \times n$ adjacency matrix~$M$ in~$\Oh(n^2)$ time. Scan through~$A$ to collect the indices of the non-zero entries in each row in~$\Oh(nm)$ time. For each row~$r$, for each of the~$\Oh(r^2)$ pairs~$A[r,i], A[r,j]$ of distinct nonzero entries in the row, set the corresponding entries~$M[i,j]$,~$M[j,i]$ of the adjacency matrix to one.
\end{proof}

We show how to obtain a nice Gaifman decomposition for a matrix~$A$ of width~$w$ from any tree decomposition of its Gaifman graph~$G(A)$ of width~$w$.

\begin{proposition} \label{proposition:gaifman:decomposition}
There is an algorithm that, given~$A \in \Z^{m \times n}$ and a width-$w$ tree decomposition~$(T,\X)$ of the Gaifman graph of~$A$, computes a nice Gaifman decomposition~$(T',\X',\ZZ')$ of~$A$ having width~$w$ in~$\Oh(w^2 \cdot |V(T)| + n \cdot m \cdot w)$ time.
\end{proposition}

\begin{proof}
\emph{Building a nice tree decomposition.} From the tree decomposition~$(T,\X)$ of~$G(A)$ we can derive a chordal supergraph of~$G(A)$ with maximum clique size bounded by~$w + 1$, by completing the vertices of each bag into a clique~\cite[Lemma 2.1.1]{Kloks94}. This can be done in~$\Oh(w^2 \cdot |V(T)|)$ time by scanning through the contents of the bags of~$(T,\X)$. A perfect elimination order of the chordal supergraph can be used to obtain a nice tree decomposition~$(T',\X')$ of~$G(A)$ having width~$w$ on at most~$4n$ nodes~\cite[Lemma 13.1.3]{Kloks94}. The nice tree decomposition consists of leaf, join, introduce, and forget nodes. 

\emph{Incorporating constraint nodes.} We augment the nice tree decomposition with constraint nodes to obtain a nice Gaifman decomposition of~$A$, as follows. We scan through matrix~$A$ and store, for each row, a list of pointers to the nonzero entries in that row. This takes~$\Oh(n \cdot m)$ time. Since a graph of treewidth~$w$ does not have cliques of size more than~$w+1$, by Observation~\ref{observation:rows:clique} each row of~$A$ has at most~$w+1$ nonzero entries. We maintain a list of the rows in~$A$ that have not yet been associated to a constraint bag in the Gaifman decomposition. We traverse the rooted tree~$T'$ in post-order. For each node~$i$, we inspect the corresponding bag~$X_i$ and test, for each constraint that is not yet represented by the decomposition, whether all variables involved in the constraint are contained in the bag. This can be determined in~$\Oh(w)$ time per constraint as follows. For each variable in~$X_i$ we test whether the corresponding row in~$A$ has a nonzero coefficient for that variable; if so, we increase a counter. If the final value of the counter matches the precomputed number of nonzero coefficients in the row then the bag contains all variables involved in the constraint. In that case we update the tree~$T'$ as follows: we make a new node~$i'$, assign~$X_{i'} := X_i$, and let~$Z'_{i'}$ be a copy of the precomputed list of pointers to the nonzero coefficients in the constraint. We make~$i'$ the parent of~$i$. If~$i$ is not the root, then it originally had a parent~$j$; we make~$j$ the parent of~$i'$ instead. This operation effectively splices a node of degree two into the tree. Since the newly introduced node has the same bag as~$i$, the relation between the bags of parents and children for the existing nodes of the tree remains unaltered (e.g., a forget node in the old tree will be a forget node in the new tree). The newly introduced node~$i'$ satisfies the requirements of a constraint node. We then continue processing the remainder of the tree to obtain the final nice Gaifman decomposition~$(T',X',\ZZ')$. As the original tree contains~$\Oh(n)$ nodes, while we spend~$\Oh(m \cdot w)$ time per node to incorporate the constraint bags, this phase of the algorithm takes~$\Oh(n \cdot m \cdot w)$ time. By Observation~\ref{observation:rows:clique}, for each constraint of~$A$ the involved variables occur together in some bag. Hence we will detect such a bag in the procedure, which results in a constraint node for the row. Since the nice tree decomposition that we started from contained at most~$4n$ nodes, while we introduce one node for each constraint in~$A$, the resulting tree has at most~$4n+m$ nodes. This shows that~$(T',X',\ZZ')$ satisfies all properties of a nice Gaifman decomposition and concludes the proof.
\end{proof}

\subsection{Feasibility on Gaifman graphs of bounded treewidth}\label{section:treewidth:basic}

We discuss the influence of treewidth on the complexity of \ILPF. It turns out that for unbounded domain variables the problem remains weakly \NP-hard on instances with Gaifman graphs of treewidth at most two (Theorem~\ref{theorem:unboundeddomain:boundedtw:hardness}). On the other hand, the problem can be solved by a simple dynamic programming algorithm with runtime~$\Oh^*(d^{w+1})$, where~$d$ is the domain size and~$w$ denotes the width of a given tree decomposition of the Gaifman graph (Theorem~\ref{theorem:boundeddomain:boundedtreewidth:dp}). In other words, the problem is fixed-parameter tractable in terms of~$d+w$, and efficiently solvable for bounded treewidth and~$d$ polynomially bounded in the input size.

Both results are not hard to prove and fixed-parameter tractability of \ILPF{}($d+w$) can also be derived from Courcelle's theorem (cf.~\cite[Corollary 11.43]{FlumG06}). Nevertheless, for the sake of self-containment and concrete runtime bounds we provide direct proofs. Theorem~\ref{theorem:boundeddomain:boundedtreewidth:dp} is a subroutine of our protrusion reduction algorithm.

\begin{theorem}\label{theorem:unboundeddomain:boundedtw:hardness}
\problem{ILP Feasibility} remains weakly \NP-hard when restricted to instances $(A,b)$ whose Gaifman graph~$G(A)$ has treewidth two.
\end{theorem}

\begin{proof}
We give a straightforward reduction from \problem{Subset Sum} to this restricted variant of \ILPF. Recall that an instance of \problem{Subset Sum} consists of a set~$\{a_1,\ldots,a_n\}$ of integers and a target value~$b\in\Z$; the task is to determine whether some subset of the~$n$ integers sums to exactly~$b$.
Given such an instance~$(\{a_1,\ldots,a_n\},b)$ we create~$n$ variables~$x_1,\ldots,x_n$ that encode the selection of a subset and~$n$ variables~$y_1,\ldots,y_n$ that effectively store partial sums; the~$x_i$ variables are constrained to domain~$\{0,1\}$. Concretely, we aim to compute
\[
y_j=\sum_{i=1}^j a_ix_i
\]
for all~$j\in\{1,\ldots,n\}$. Clearly, this is correctly enforced by the following constraints.
\begin{align*}
y_1&=a_1x_1\\
y_j&=a_jx_j+y_{j-1}&j\in\{2,\ldots,n\}
\end{align*}
Finally, we enforce~$y_n=b$. Clearly, a subset of the~$n$ integers with sum~$b$ translates canonically to a feasible assignment of the variables, and vice versa.

It remains to check the treewidth of the corresponding Gaifman graph. We note that for this purpose it is not necessary to split equalities into inequalities or performing similar normalizations since it does not affect whether sets of variables occur in at least one shared constraint. Thus, we can make a tree decomposition (in fact, a path decomposition) consisting of a path on nodes~$1,\ldots,n$ with bags~$X_1,\ldots,X_n$ where~$X_1=\{x_1,y_1\}$ and~$X_j=\{x_j,y_{j-1},y_j\}$ for~$j\geq 2$. Clearly, for each constraint there is a bag containing all its variables, correctly handling all edges of the Gaifman graph, and the bags containing any variable form a connected subtree. It follows that the Gaifman graph of the constructed instance has treewidth at most two.
\end{proof}

\begin{theorem}\label{theorem:boundeddomain:boundedtreewidth:dp}
Instances~$(A\in\Z^{m\times n},b)$ of \ILPF of domain size~$d$ with a given nice Gaifman decomposition of width~$w$ can be solved in time $\Oh(d^{w+1} \cdot w \cdot (n+m))$.
\end{theorem}

\begin{proof}
Let~$(A,b)$ denote an instance of \ILPF of domain size~$d$ and let~$(T,\X,\ZZ)$ denote a given nice Gaifman decomposition of width~$w$ for~$A$. We describe a simple dynamic programming algorithm for testing whether there exists an integer vector~$x$ such that~$Ax\leq b$. For ease of presentation we assume that each domain is~$\{0,\ldots,d-1\}$; it is straightforward, but technical, to use arbitrary (possibly different) domains of at most~$d$ values for each variable.

With a node~$i$, apart from its bag~$X_i$, we associate the set~$V_i$ of all variables appearing in~$X_i$ or in the bag~$X_j$ of any descendant~$j$ of~$i$. By~$C_i$ we denote all constraints~$p$ (rows of~$A$) that the Gaifman decomposition maps to a descendant of~$i$ (including~$i$ itself).

Our goal is to compute for each node~$i$ the set of all feasible assignments to the variables in~$X_i$ \emph{when taking into account only constraints in~$C_i$}. The set of feasible assignments for~$X_i$ will be recorded in a table~$F_i$ indexed by tuples~$\{0,\ldots,d-1\}^\ell$ where~$\ell=|X_i|$. The entry at~$F_i(a_1,\ldots,a_\ell)$ corresponds to assigning~$x_{i_1}=a_1,\ldots,x_{i_\ell}=a_\ell$ where~$X_i=\{x_{i_1},\ldots,x_{i_\ell}\}$ and~$i_1<\ldots<i_\ell$. Its value will be~$1$ if we determined that there is a feasible assignment for~$V_i$ with respect to constraints~$C_i$ that extends~$x_{i_1}=a_1,\ldots,x_{i_\ell}=a_\ell$; otherwise the value is~$0$. We will now describe how to compute the tables~$F_i$ in a bottom-up manner, by outlining the behavior on each node type.
\begin{itemize}
 \item \emph{Leaf node~$i$ with bag~$X_i=\{x_q\}$.} Since~$C_i = \emptyset$ as~$i$ has no children, we simply have~$F_i(a)=1$ for all~$a\in\{0,\ldots,d-1\}$ which is computed in~$\Oh(d)$ time.
 
 \item \emph{Forget node~$i$ with child~$j$ and bag~$X_i=X_j\setminus\{x_q\}$.} It suffices to project~$F_j$ down to only contain information about the feasible assignments for~$X_i=X_j\setminus \{x_q\}$. To make this concrete, let~$X_j=\{x_{i_1},\ldots,x_{i_\ell}\}$ with~$\ell=|X_j|=|X_i|+1$ and~$x_q=x_{i_s}$. Thus,~$X_i=\{x_{i_1},\ldots,x_{i_{s-1}},x_{i_{s+1}},\ldots,x_{i_\ell}\}$. We let
 \[
  F_i(a_1,\ldots,a_{s-1},a_{s+1},\ldots,a_\ell)=\max_{a\in\{0,\ldots,d-1\}}F_j(a_1,\ldots,a_{s-1},a,a_{s+1},\ldots,a_\ell),
 \]
 for all~$(a_1,\ldots,a_{s-1},a_{s+1},\ldots,a_\ell)\in\{0,\ldots,d-1\}^{|X_i|}$, i.e., we set~$F_i(\ldots)$ to~$1$ if some choice of~$a$ extends the assignment to one that is feasible for~$X_j$ with respect to all constraints on~$V_j=V_i$; else it takes value~$0$. Each table entry takes time~$\Oh(d)$ and in total we spend time~$\Oh(d\cdot d^{|X_i|})\subseteq\Oh(d^{w+1})$; note that~$X_i=X_j\setminus\{x_q\}$ implies~$|X_i|\leq w$.
 
 \item \emph{Introduce node~$i$ with child~$j$ and bag~$X_i=X_j\cup\{x_q\}$.} Let~$X_i=\{x_{i_1},\ldots,x_{i_\ell}\}$ with~$\ell=|X_i|$ and~$x_q=x_{i_s}$, implying that~$X_j=\{x_{i_1},\ldots,x_{i_{s-1}},x_{i_{s+1}},\ldots,x_{i_\ell}\}$. As~$i$ is an introduce node we have~$C_i = C_j$, and therefore~$C_i$ does not constrain the value of~$x_q$ in any way. We set~$F_i$ as follows.
 \[
  F_i(a_1,\ldots,a_\ell)=F_j(a_1,\ldots,a_{s-1},a_{s+1},\ldots,a_\ell),
 \]
 for all~$(a_1,\ldots,a_\ell)\in\{0,\ldots,d-1\}^{|X_i|}$. We use time~$\Oh(d^{|X_i|})\subseteq \Oh(d^{w+1})$.
 
 \item \emph{Join node~$i$ with children~$j,j'$ and bag~$X_i=X_j=X_{j'}$.} At a join node~$i$, from child~$j$ we get all assignments that are feasible for~$X_i=X_j$ regarding constraints~$C_j$, and from child~$j'$ we get the feasible assignments for constraints~$C_{j'}$. We have~$C_i = C_j \cup C_{j'}$, and therefore an assignment is feasible for~$C_i$ if and only if it is feasible for both~$C_j$ and~$C_{j'}$. It suffices to merge the information of the two child nodes. Letting~$X_i=\{x_{i_1},\ldots,x_{i_\ell}\}$ with~$\ell=|X_i|$, we set
 \[
  F_i(a_1,\ldots,a_\ell)=\min\{F_j(a_1,\ldots,a_\ell),F_{j'}(a_1,\ldots,a_\ell)\},
 \]
 for all~$(a_1,\ldots,a_\ell)\in\{0,\ldots,d-1\}^{|X_i|}$. This takes time~$\Oh(d^{|X_i|})\subseteq\Oh(d^{w+1})$.

	\item \emph{Constraint node~$i$ with child~$j$ mapped to row~$p$.} Let~$X_i=\{x_{i_1},\ldots,x_{i_\ell}\}$ with~$\ell=|X_i|$. We know that~$C_i = C_j \cup \{p\}$ and therefore an assignment of values to~$X_i$ is feasible with respect to the constraints~$C_i$ if and only if it is feasible for the constraints~$C_j$ and also satisfies constraint~$p$. We therefore initialize by setting~$F_i$ as follows.
 \[
  F_i(a_1,\ldots,a_\ell)=F_j(a_1,\ldots,a_\ell),
 \]
 for all~$(a_1,\ldots,a_\ell)\in\{0,\ldots,d-1\}^{|X_i|}$. Now, we need to discard those assignments~$(a_1,\ldots,a_\ell)$ that are not feasible for the additional constraint~$p$. For each~$(a_1,\ldots,a_\ell)$ with~$F_i(a_1,\ldots,a_\ell)=1$ we process the row~$p$. Using the pointers to the nonzero coefficients in row~$p$ that are stored in~$Z_i$, along with the fact that all variables constrained by~$p$ are contained in bag~$X_i$ by Definition~\ref{definition:gaifman:decomposition}, we can evaluate the constraint in~$\Oh(w)$ time. If the sum of values times coefficients exceeds~$b[p]$ then the constraint is not satisfied and we set~$F_i(a_1,\ldots,a_\ell)$ to~$0$. This takes time~$\Oh(|X_i|)\subseteq\Oh(w)$ per assignment. In total we use time~$\Oh(w \cdot d^{|X_i|})\subseteq \Oh(d^{w+1}w)$.
\end{itemize}

At the end of the computation we have the table~$F_r$ where~$r$ denotes the root of the tree decomposition~$(T,\X)$. By definition, it encodes all assignments to~$X_r$ that can be extended to assignments that are feasible for all constraints in~$C_r$. By Definition~\ref{definition:gaifman:decomposition} the set~$C_r$ for the root contains all constraints of~$A$ and thus any entry~$F_r(a_1,\ldots,a_{|X_r|})=1$ implies that~$Ax\leq b$ has an integer solution. Conversely, any integer solution must lead to a nonzero entry in~$F_r$. By Definition~\ref{definition:gaifman:decomposition} the number of nodes in~$T$ is~$\Oh(n+m)$. The total time needed for the dynamic programming is therefore bounded by~$\Oh(d^{w+1} \cdot w \cdot (n+m))$.
\end{proof}

If a nice Gaifman decomposition is not given, one can be computed by combining an algorithm for computing tree decompositions~\cite{Bodlaender96_lintimetw,BodlaenderDDFLP13} with Proposition~\ref{proposition:gaifman:decomposition}.

\subsection{Protrusion reductions}\label{section:treewidth:protrusionreduction}
To formulate the protrusion replacement rule, which is the main algorithmic asset used in this section, we need some terminology. For a non-negative integer~$r$, an \emph{$r$-boundaried ILP} is an instance~$(A,b)$ of \ILPF in which~$r$ distinct \emph{boundary variables}~$x_{t_1}, \ldots, x_{t_r}$ are distinguished among the total variable set~$\{x_1, \ldots, x_n\}$. If~$Y = (x_{i_1}, \ldots, x_{i_r})$ is a sequence of variables of~$Ax \leq b$, we will also use~$(A,b,Y)$ to denote the corresponding $r$-boundaried ILP. The \emph{feasible boundary assignments} of a boundaried ILP are those assignments to the boundary variables that can be completed into an assignment that is feasible for the entire system.

\begin{definition}
Two $r$-boundaried ILPs~$(A,b,x_{t_1},\ldots,x_{t_r})$ and~$(A',b',x'_{t'_1},\ldots,x'_{t'_r})$ are \emph{equivalent} if they have the same feasible boundary assignments.
\end{definition}

The following lemma shows how to compute equivalent boundaried ILPs for any boundaried input ILP. The replacement system is built by adding, for each infeasible boundary assignment, a set of constraints on auxiliary variables that explicitly blocks that assignment.

\begin{lemma}\label{lemma:tw:boundariedilp:replacement}
There is an algorithm with the following specifications: (1) It gets as input an~$r$-boundaried ILP~$(A,b,x_{t_1}, \ldots, x_{t_r})$ with domain size~$d$, with~$A\in\Z^{m\times n}$,~$b\in\Z^m$, and a width-$w$ nice Gaifman decomposition~$(T,\X,\ZZ)$ of~$A$. (2) Given such an input it takes time~$\Oh(d^r \cdot (d^{w+1}w (n+m)) + r^2 \cdot d^{2r})$. (3) Its output is an equivalent~$r$-boundaried ILP~$(A',b',x'_{t'_1},\ldots,x'_{t'_r})$ of domain size~$d$ containing~$\Oh(r \cdot d^r)$ variables and constraints, and all entries of~$(A',b')$ in~$\{-d,\ldots,d\}$.
\end{lemma}

\begin{proof}
The lemma is a combination of two ingredients. Using Theorem~\ref{theorem:boundeddomain:boundedtreewidth:dp} we can efficiently test whether a given assignment to the boundary variables can be extended to a feasible solution for~$(A,b)$. Then, knowing the set of all assignments that can be feasibly extended, we block each infeasible partial assignment by introducing a small number of variables and constraints, allowing us to fully discard the original constraints and non-boundary variables. The latter step uses a construction from an earlier paper by Kratsch~\cite[Theorem 5]{Kratsch13_ilp1}, which we repeat here for completeness.

\emph{Finding feasible partial assignments.} Consider a partial assignment~$x_{t_1} = a_1, \ldots, x_{t_r} = a_r$ to the boundary variables with~$a_i \in \{0, \ldots, d-1\}$. To determine whether this assignment can be extended to a feasible assignment for~$(A,b)$, we enforce these equalities in the ILP. Concretely, for each~$i$ we replace the domain-bounding constraints for~$x_{t_i}$ in the system~$(A,b)$ by~$- x_i \geq a_i$ and~$x_i \leq a_i$. We obtain a new system~$(\hat{A}, \hat{b})$ with~$m$ constraints. Since the modified constraints involve only a single variable each, the modifications do not affect the Gaifman graph:~$G(A) = G(\hat{A})$. Moreover, the modifications do not affect which entries in the constraint matrix have nonzero values, implying that~$(T,\X,\ZZ)$ also serves as a nice Gaifman decomposition of~$\hat{A}$. The partial assignment to the boundary variables can be feasibly extended if and only if~$(\hat{A}, \hat{b})$ is feasible. We may invoke the algorithm of Theorem~\ref{theorem:boundeddomain:boundedtreewidth:dp} with~$(T,\X,\ZZ)$ to decide feasibility in~$\Oh(d^{w+1} w (n+m))$ time. By iterating over all~$d^r$ possible partial assignments to the boundary variables with values in~$\{0, \ldots, d-1\}$, we determine which partial assignments can be feasibly extended. Let~$\L$ be a list of the partial assignments that \emph{can not} be extended to a feasible solution for~$(A,b)$.

\emph{Blocking infeasible partial assignments.} Using~$\L$ we construct an equivalent $r$-boundaried ILP~$(A',b',x'_{t_1},\ldots,x'_{t_r})$ as follows. Based on the length of~$\L$ we can determine the number of variables that will be used in~$(A',b')$, which helps to write down the constraint matrix efficiently. The number of variables in the new system will be~$r + 2r |\L|$, the number of constraints will be~$2r + (6r+1) |\L|$. The system is built as follows. For each boundary variable~$x_{t_i}$ of~$(A,b)$ we introduce a corresponding variable~$x'_{t'_i}$ in~$(A',b')$ and constrain its domain to~$\{0, \ldots, d-1\}$ using two inequalities; this yields~$r$ variables and~$2r$ constraints. 

For each infeasible partial assignment~$a^j = (a^j_1, \ldots, a^j_r)$ in the list~$\L$, we add new variables~$u^j_i$ and~$v^j_i$ for all~$i \in [r]$, together with the following constraints:

\begin{align}
\forall i\in [r]:&& x'_{t_i} &= a^j_i+ u^j_i- d\cdot v^j_i && u^j_i\in\{0,\ldots,d-1\}, v^j_i\in\{0,1\}\label{constraint:checkequal}\\
&&\sum_{i=1}^r u^j_i&\geq 1 \label{constraint:checkdummies}
\end{align}

We claim that an assignment to the boundary variables can be extended to the newly introduced variables to satisfy the constraints for~$a^j$ if and only if the partial assignment is not~$a^j$. In the first direction, assume that~$(x'_{t_1},\ldots,x'_{t_r})=(a^j_1,\ldots,a^j_r)$. Then~$0=x'_{t_i}-a^j_i=u^j_i-d\cdot v^j_i$, implying that~$u^j_i=v^j_i=0$ (taking into account the domains of~$u^j_i$ and~$v^j_i$) for all~$i$. Therefore constraint~\cref{constraint:checkdummies} is violated which shows that the partial assignment can not be feasibly extended. In the other direction, if~$(x'_{t_1},\ldots,x'_{t_r})\neq (a^j_1,\ldots,a^j_r)$, then there is a position~$i$ with~$x'_{t_i} \neq a^j_i$. It follows that~$0<|x'_{t_i}-a^j_i|<d$ (due to the domain of~$x'_{t_i}$) which in turn implies that~$u^j_i\neq 0$ since the contribution of~$d\cdot v^j_i$ to the equality~\cref{constraint:checkequal} is a multiple of~$d$. Therefore constraint~\cref{constraint:checkdummies} is fulfilled.

The only coefficients used in the constraints that block a partial assignment are~$\{-d, -1, 0, 1, d\}$; as the equalities of~\cref{constraint:checkequal} are represented using two inequalities, we get coefficients~$+d$ and~$-d$. The values of~$a^j_i$, which appear in the right-hand side vector~$b'$, are in~$\{-(d-1),\ldots,0, \ldots, d-1\}$ since they arise from the coordinates of an attempted partial assignment (with negative values from representing equalities by inequalities). As the coefficients of the domain-enforcing constraints are all plus or minus one, with right-hand side in~$\{-(d-1),\ldots,0, \ldots, d-1\}$, the structure of the constructed system~$(A',b')$ matches that described in the lemma statement. For each infeasible partial assignment we introduce~$2r$ variables with~$2$ domain-enforcing constraints each, along with~$2r$ inequalities to express the~$r$ equalities of~\cref{constraint:checkequal}, and a single constraint for~\cref{constraint:checkdummies}. The total system therefore has~$r + 2r |\L|$ variables and~$2r + (6r + 1) |\L|$ constraints. Since there are only~$d^r$ partial assignments that we check, we have~$|\L| \leq d^r$ and therefore the system has~$\Oh(r \cdot d^r)$ variables and constraints. Consequently, the constraint matrix~$A'$ has~$\Oh(r^2 \cdot d^{2r})$ entries and it is not hard to verify that it can be written down in linear time. It remains to prove that the $r$-boundaried ILP~$(A',b',x'_{t'_1}, \ldots, x'_{t'_r})$ is equivalent to the input structure.

Consider an assignment to the boundary variables. If the assignment can be extended to a feasible assignment for~$(A,b)$, then the boundary variables take values in~$\{0, \ldots, d-1\}$ (since~$(A,b)$ has domain size~$d$) and therefore satisfy the domain restrictions of~$(A',b')$. Since the partial assignment has a feasible extension, it is not on the list~$\L$. For each set of constraints that was added to block an infeasible partial assignment~$a^j$, the claim above therefore shows that the related variables~$u^j_i$ and~$v^j_i$ can be set to satisfy their constraints. Hence the partial assignment can be extended to a feasible assignment for~$(A',b')$. In the reverse direction, suppose that a partial assignment can be feasibly extended for~$(A',b')$. By the claim above, the partial assignment differs from each of the blocked points on~$\L$. Since~$\L$ contains all infeasible partial assignments with values in~$\{0, \ldots, d-1\}^r$, and feasible partial assignments for~$(A',b')$ belong to~$\{0, \ldots, d-1\}^r$ since we restricted the domain of~$x'_{t'_1}, \ldots, x'_{t'_r}$, there is an extension feasible for~$(A,b)$. This shows that the two $r$-boundaried ILPs are indeed equivalent, and concludes the proof.
\end{proof}

Intuitively, we can simplify an \ILPF instance $(A,b)$ with a given protrusion decomposition by replacing all protrusions with equivalent boundaried ILPs of small size via Lemma~\ref{lemma:tw:boundariedilp:replacement}. We get a new instance containing all replacement constraints plus all original constraints that are fully contained in the shared part.

\begin{theorem} \label{theorem:shrink:protrusion:decomposition}
For each constant~$r$ there is an algorithm that, given an instance $(A,b)$ of \ILPF with domain size~$d$, along with a~$(k,r)$-protrusion decomposition~$Y_0 \dotcup Y_1 \dotcup \ldots \dotcup Y_\ell$ of the given Gaifman graph~$G(A)$, outputs an equivalent instance~$(A',b')$ of \ILPF with domain size~$d$ on~$\Oh(k \cdot d^r)$ variables in time~$\Oh(n \cdot m + d^{2r}(n + m) + k \cdot m \cdot d^{r} + k^2 \cdot d^{2r})$. Each constraint of~$(A',b')$ is either a constraint in~$(A,b)$ involving only variables from~$Y_0$, or one of~$\Oh(k \cdot d^r)$ new constraints with coefficients and right-hand side among~$\{-d, \ldots, d\}$.
\end{theorem}

\begin{proof}
The main idea of the proof is to apply Lemma~\ref{lemma:tw:boundariedilp:replacement} to replace each protrusion in the Gaifman graph by a small subsystem that is equivalent with respect to the boundary variables. For the sake of efficiency, we start by scanning through~$A$ once to compute for each row of~$A$ a list of pointers to the nonzero coefficients in that row. This takes~$\Oh(n \cdot m)$ time. We handle the protrusions~$Y_i$ for~$i \in [\ell]$ (this implies~$i \geq 1$) consecutively, iteratively replacing each protrusion by a small equivalent boundaried ILP to obtain an equivalent instance. 

\emph{Replacing protrusions.} Consider some~$Y_i$ with~$i \geq 1$; we show how to replace the variables~$Y_i$ by a small subsystem that has the same effect on the existence of global solutions. The definition of protrusion decomposition ensures that~$N_{G(A)}(Y_i) \subseteq Y_0$ and that~$G_i := G(A)[Y_i \cup N_{G(A)}(Y_i)]$ has treewidth at most~$r - 1$. From the system~$(A,b)$ we extract the constraints involving at least one variable in~$Y_i$. We collect these constraints, together with domain-enforcing constraints for~$N_{G(A)}(Y_i) \cap Y_0$, into a subsystem~$(A_i,b_i)$. Let~$n_i$ and~$m_i$ be the number of variables and constraints in~$(A_i,b_i)$, respectively. Since~$|N_{G(A)}(Y_i) \cap Y_0| \leq r$ by the definition of a protrusion decomposition, we have~$n_i \leq |Y_i| + r$. Since the nonzero variables involved in a constraint induce a clique in the Gaifman graph, while~$G_i$ has treewidth at most~$r - 1$ and therefore does not have cliques of size more than~$r$, it follows that a constraint involving a variable from~$Y_i$ acts on at most~$r$ variables. We can therefore identify the constraints involving a variable in~$Y_i$ by only inspecting the rows of~$A$ containing at most~$r$ nonzero entries, which implies that the system~$(A_i,b_i)$ can be extracted in~$\Oh(m \cdot r + n_i \cdot m_i) \subseteq \Oh(m + n_i \cdot m_i)$ time.

Let~$\{x_{t_{i,1}}, \ldots, x_{t_{i,r' \leq r}}\}$ be the neighbors of~$Y_i$ in~$G(A)$, i.e., the variables of~$Y_0$ that appear in a common constraint with a variable in~$Y_i$. As~$r$ is a constant, we can compute a tree decomposition~$(T_i, \X_i = \{X_{i,j} \mid j \in V(T_i) \})$ of~$G_i$ of width~$r-1$ with~$\Oh(n_i)$ bags in~$\Oh(n_i)$ time~\cite{Bodlaender96_lintimetw,BodlaenderDDFLP13}. Using Proposition~\ref{proposition:gaifman:decomposition} this yields a nice Gaifman decomposition~$(T'_i,\X'_i,\ZZ'_i)$ of~$A_i$ in~$\Oh(n_i \cdot m_i)$ time. Interpreting~$(A_i,b_i, x_{t_{i,1}}, \ldots, x_{t_{i,r'}})$ as an $r'$-boundaried ILP, we invoke Lemma~\ref{lemma:tw:boundariedilp:replacement} to compute an equivalent $r'$-boundaried ILP~$(A'_i,b'_i,x'_{t'_{i,1}}, \ldots, x'_{t'_{i,r'}})$ in~$\Oh(d^{2r} (n_i + m_i))$ time for constant~$r$. By Lemma~\ref{lemma:tw:boundariedilp:replacement}, the numbers in the system~$(A'_i,b'_i)$ are restricted to the set~$\{-d, \ldots, d\}$ and~$(A'_i,b'_i)$ has~$\Oh(d^{r'}) \subseteq \Oh(d^r)$ variables and constraints. We modify the instance~$(A,b)$ as follows, while preserving the fact that it has domain size~$d$. We remove all variables from~$Y_i$ and all constraints involving them from the system~$(A,b)$. For each non-boundary variable in~$(A'_i,b'_i)$ we add a corresponding new variable to~$(A,b)$. For each constraint in~$(A'_i,b'_i)$, containing some boundary variables and some non-boundary variables, we add a new constraint with the same coefficients and right-hand side to~$(A,b)$. All occurrences of boundary variables~$x'_{t'_{i,j}}$ of~$(A'_i,b'_i)$ are replaced by the corresponding existing variables~$x_{t_{i,j}}$ of~$(A,b)$; occurrences of non-boundary variables are replaced by occurrences of the corresponding newly introduced variables. 

Observe that these replacements preserve the variable set~$Y_0$, and that the newly introduced constraints only involve~$Y_0$ and newly introduced variables. We can therefore perform this replacement step independently for each protrusion~$Y_i$ with~$i \in [\ell]$. Since each variable set~$Y_i$ for~$i \in [\ell]$ is removed and replaced by a new set of~$\Oh(d^{r})$ variables, the final system~$(A',b')$ resulting from these replacements has~$\Oh(|Y_0| + \ell \cdot d^{r})$ variables, which is~$\Oh(k \cdot d^{r})$ since the definition of a $(k,r)$-protrusion decomposition ensures that~$\max(\ell, |Y_0|) \leq k$. When building~$(A',b')$, the procedure above removes from~$(A,b)$ all constraints that involve at least one variable in~$Y_i$ with~$i \geq 1$. Hence the only constraints in~$(A',b')$ are (1) the constraints of~$(A,b)$ that only involve variables in~$Y_0$, and (2) the~$\Oh(\ell \cdot d^{r}) \subseteq \Oh(k \cdot d^{r})$ new constraints that are copied from subsystems~$(A'_i,b'_i)$ for~$i \in [\ell]$. Hence the constraints of the constructed instance satisfy the claims in the theorem statement.

\emph{Running time.} Let us consider the running time of the procedure. For each~$i \in [\ell]$ the time to extract the subsystem~$(A_i,b_i)$ and compute an equivalent $r$-boundaried ILP is dominated by~$\Oh(m + n_i \cdot m_i + d^{2r}(n_i + m_i))$, using the fact that~$r$ is a constant. As we observed earlier,~$n_i \leq |Y_i| + r$. Each constraint in~$(A_i,b_i)$ is either a constraint of~$A$ involving a variable in~$Y_i$, or a domain-enforcing constraint on~$N_{G(A)}(Y_i) \cap Y_0$; there are at most~$2r$ of the latter kind. The constraints of~$A$ involving a variable in~$Y_i$ do not occur in systems~$(A_{i'}, b_{i'})$ for~$i' \neq i$, since this would imply that a variable in~$Y_i$ is adjacent in~$G(A)$ to a variable in~$Y_{i'}$ for~$i \neq i' \neq 0$, contradicting the definition of a protrusion decomposition. Since all but~$2r$ of the constraints of~$(A_i,b_i)$ correspond to constraints of~$A$, which only occur once, it follows that~$\sum _{i=1}^\ell m_i \leq m + \ell \cdot 2r \leq m + n \cdot 2r \leq m + m \cdot 2r \in \Oh(m)$, using the fact that~$\ell \leq n$ by the definition of protrusion decomposition and~$n \leq m$ since~$(A,b)$ contains a domain-enforcing constraint for each variable. Similarly we have~$\sum _{i=1}^\ell n_i \in \Oh(n)$, which shows that~$\sum _{i=1}^\ell (n_i \cdot m_i) \leq (\sum _{i=1}^\ell n_i) \cdot (\sum _{i=1}^\ell m_i) \in \Oh(n \cdot m)$. From this it follows that, summing the total running time of computing replacements for all protrusions~$Y_i$ with~$i \in [r]$, we obtain the bound~$\sum _{i=1}^\ell \Oh(m + n_i \cdot m_i + d^{2r} (n_i + m_i)) \subseteq \Oh(n \cdot m + d^{2r}(n + m))$, using again that~$\ell \cdot m \leq n \cdot m$. After all replacements have been computed, the time to construct the output instance~$(A',b')$ is dominated by the total size of the resulting matrix~$A'$. Let~$n'$ and~$m'$ be the number of variables and constraints in~$A'x' \leq b'$. The argumentation above shows that~$n' \in \Oh(k \cdot d^{r})$. Each constraint in~$A'$ that is not in~$A$ originates from a subsystem~$(A'_i,b'_i)$ which has~$\Oh(d^{r})$ constraints. Hence~$m' \in \Oh(m + \ell \cdot d^{r})$. The number of entries in~$A'$ is therefore~$\Oh((k \cdot d^{r}) \cdot (m + \ell \cdot d^{r}))$. Using the fact that~$\ell \leq k$ by the definition of a protrusion decomposition, this simplifies to~$\Oh(k \cdot m \cdot d^{r} + k^2 \cdot d^{2r})$, which also bounds the time to write down the system~$(A',b')$. (We remark that these bounds can be improved using a sparse matrix representation.) The total running time of the procedure is therefore~$\Oh(n \cdot m + d^{2r}(n + m) + k \cdot m \cdot d^{r} + k^2 \cdot d^{2r})$. Having proven the running time bound, it remains to show that the two instances of \ILPF are equivalent.

\begin{claim}
There is an integer vector~$x$ satisfying~$Ax \leq b$ if and only if there is an integer vector~$x'$ satisfying~$A' x' \leq b'$. 
\end{claim}
\begin{proof}
($\Rightarrow$) In the first direction, assume that~$Ax \leq b$. Then there is a partial assignment~$x_{Y_0}$ of values to the variables~$Y_0$ that can be extended to a feasible solution for~$(A,b)$, as~$x$ is such an extension. Since all variables~$Y_0$ also exist in~$(A',b')$ we can consider~$x_{Y_0}$ as a partial assignment of variables for~$(A',b')$. We prove that this partial assignment~$x_{Y_0}$ can be extended to a feasible assignment for~$(A',b')$. To see this, observe that all constraints in~$(A',b')$ involving only variables of~$Y_0$ also exist in~$(A,b)$, and are therefore satisfied by the partial assignment. All constraints involving at least one variable from~$Y_i$ for~$i \geq 1$ were removed and replaced by constraints from a subsystem~$(A'_i,b'_i)$. Consider a subsystem~$(A'_i,b'_i)$ whose constraints and variables were introduced into~$(A',b')$. Letting~$r' := |N_{G(A)}(Y_i) \cap Y_0|$, Lemma~\ref{lemma:tw:boundariedilp:replacement} guarantees that the $r'$-boundaried ILP~$(A_i,b_i,N_{G(A)}(Y_i) \cap Y_0)$ is equivalent to the $r'$-boundaried ILP~$(A'_i,b'_i,N_{G(A)}(Y_i) \cap Y_0)$. Since~$Ax \leq b$, and~$(A_i, b_i)$ is a subsystem of~$(A,b)$, the partial assignment of~$x_{Y_0}$ to the variables in~$N_{G(A)}(Y_i) \cap Y_0$ can be extended to a feasible solution of~$(A_i,b_i)$. By definition of equivalence, this shows that the partial assignment can be extended to a feasible solution for~$(A'_i,b'_i,N_{G(A)}(Y_i))$, implying that the new variables that were added to~$(A',b')$ originating from~$(A'_i,b'_i)$ can be assigned values from~$\{0, \ldots, d-1\}$ to satisfy the constraints in~$(A'_i,b'_i)$. Since these are the only constraints involving the variables that were added from~$(A'_i,b'_i)$, we can independently assign values to the new variables introduced by the subsystems~$(A'_i,b'_i)$ for~$i \in [\ell]$ to obtain a feasible assignment for~$(A',b')$.

($\Leftarrow$) For the reverse direction, assume that there is an integer vector~$x'$ such that~$A' x' \leq b'$ and consider the corresponding partial assignment~$x'_{Y_0}$ to the variables~$Y_0$. We will prove that~$x'_{Y_0}$ can be extended to a feasible assignment for~$(A,b)$. As~$(A',b')$ contains all constraints of~$(A,b)$ whose variable sets are contained in~$Y_0$, we only have to show that~$x'_{Y_0}$ can be extended to the variables~$Y_1, \ldots, Y_\ell$ while satisfying all constraints involving at least one variable not in~$Y_0$. Observe that for each variable~$x_j \not \in Y_0$, all variables that occur in a common constraint with~$x_j$ are neighbors of~$x_j$ in~$G(A)$. Consequently, if we choose~$i \in [\ell]$ such that~$x_j \in Y_i$, then all variables that occur in a constraint with~$x_j$ are included in the set~$Y_i \cup (N_{G(A)}(Y_i) \cap Y_0)$, which implies that all constraints involving~$x_j$ are included in the subsystem~$(A_i,b_i)$. Note that~$(A',b')$ contains the boundaried ILP~$(A'_i,b'_i,N_{G(A)}(Y_i) \cap Y_0)$ which is equivalent to the boundaried ILP~$(A_i,b_i,N_{G(A)}(Y_i) \cap Y_0)$ that is included in~$(A,b)$. As~$x'_{Y_0}$ can be extended to a feasible assignment for~$(A',b')$, this partial assignment can be extended to a feasible assignment for~$(A'_i,b'_i,N_{G(A)}(Y_i) \cap Y_0)$, which implies by the definition of equivalence that it can also be extended to a feasible assignment for~$(A_i,b_i,N_{G(A)}(Y_i) \cap Y_0)$, which includes all variables in~$Y_i$ and all constraints involving at least one variable of~$Y_i$. Hence for each~$i \in [\ell]$ we can extend~$x'_{Y_0}$ to the variables of~$Y_i$ while satisfying the constraints in~$(A_i,b_i)$. These extensions can be done independently since each variable is present in only one set~$Y_i$. As this shows that~$x'_{Y_0}$ can be extended to the remaining variables of~$(A,b)$ to satisfy all constraints that involve at least one variable not in~$Y_0$, it follows that~$(A,b)$ has a feasible integer solution by the argument given above. This shows that~$(A,b)$ and~$(A',b')$ are equivalent.
\end{proof}

This concludes the proof of Theorem~\ref{theorem:shrink:protrusion:decomposition}.
\end{proof}

\subsection{Limitations for replacing protrusions}\label{section:tw:limitations}

In this section, we discuss limitations regarding the replacement of protrusions in an ILP. First of all, there is an information-theoretic limitation for the worst-case size replacement of any~$r$-boundaried ILP with variables~$x_{t_1},\ldots,x_{t_r}$ each with domain size~$d$. Clearly, there are~$d^r$ different assignments to the boundary. For any set~$A$ of assignments to the boundary variables, using auxiliary variables and constraints one can construct an $r$-boundaried ILP whose feasible boundary assignments are exactly~$A$. This gives a lower bound of~$d^r$ bits for the encoding size of a general $r$-boundaried ILP, since we have~$2^{d^r}$ subsets. Our first result regarding limitations for replacing protrusions is that this lower bound even holds for boundaried ILPs of bounded treewidth.

\begin{proposition} \label{proposition:tw:protrusion:lowerbound}
For any~$d, r \in \mathbb{N}$ and~$A \subseteq \{0,\ldots, d-1\}^r$ there is an $r$-boundaried ILP of treewidth~$3r$ with domain size~$d$, whose feasible boundary assignments are~$A$.
\end{proposition}

The proposition follows from the fact that the encoding in Lemma~\ref{lemma:tw:boundariedilp:replacement} produces a boundaried ILP of treewidth at most~$3r$. To find an $r$-boundaried ILP of small treewidth whose feasible assignments are~$A$, we may therefore first construct an arbitrary $r$-boundaried ILP whose feasible boundary assignments are~$A$, and then invoke Lemma~\ref{lemma:tw:boundariedilp:replacement}. Our used encoding in Lemma~\ref{lemma:tw:boundariedilp:replacement} uses size~$\tilde\Oh(d^{2r})$. Note that, when using an encoding for sparse matrices, our replacement size comes fairly close to the information-theoretic lower bound.

Second, the lower bounds for $0/1$-\ILPF{}$(n)$, which follow, e.g., from lower bounds for \problem{Hitting Set} parameterized by ground set size, imply that there is no hope for a kernelization just in terms of deletion distance to a system of bounded treewidth. (This distance is upper bounded by $n$.) Note that the bound relies on a fairly direct formulation of \problem{Hitting Set} instances as ILPs, which creates huge cliques in the Gaifman graph when expressing sets as large inequalities (over indicator variables). The lower bound can be strengthened somewhat by instead representing sets less directly: For each set, ``compute'' the sum of its indicator variables using auxiliary variables for partial sums. Similarly to the example of \problem{Subset Sum}, this creates a structure of bounded treewidth. Note, however, that this structure is not a (useful) protrusion because its boundary can be as large as $n$; this is indeed the crux of having only a modulator to bounded treewidth but no guarantee for (or means of proving of) the existence of protrusions with small boundaries.

Finally, we prove in the following theorem that the mentioned information-theoretic limitation also affects the possibility of strong preprocessing, rather than being an artifact of the definition of equivalent boundaried ILPs. In other words, there is a family of \ILPF instances that already come with a protrusion decomposition, and with a single variable of large domain, but that cannot be reduced to size polynomial in the parameters of this decomposition. Note that this includes all other ways of handling these instances, which establishes that protrusions with even a single large domain boundary variable can be the crucial obstruction from achieving a polynomial kernelization.

\begin{theorem} \label{theorem:lb:largedomainvar}
Assuming \ncontainment, there is no polynomial-time algorithm that compresses instances~$(A \in \Z^{m \times n},b\in\Z^m)$ of \ILPF with entries in~$\{-n, \ldots, n\}$ that consist of $\{0,1\}$-variables except for a single variable of domain~$d \leq n$, which are given together with a $(k,5)$-protrusion decomposition~$Y_0 \dotcup Y_1 \dotcup \ldots \dotcup Y_\ell$ of~$V(G(A))$, to size polynomial in~$k + \hat{m} + \log d$, where~$\hat{m}$ is the number of constraints that affect only variables of~$Y_0$.
\end{theorem}

\begin{proof}
We prove the theorem by giving a cross-composition from \IS to an appropriate instance of \ILPF. We begin by picking a polynomial equivalence relation \eqvr on instances~$(G,k)$ of \IS. (Recall that such an instance asks whether~$G$ contains an independent set of size at least~$k$.) We let all malformed instances be equivalent, including those where~$k>|V(G)|$. We let any two well-formed instances~$(G_1,k_1)$ and $(G_2,k_2)$ be equivalent if~$|V(G_1)|=|V(G_2)|$ and~$k_1=k_2$, i.e., they have the same number of vertices and ask for the same size of independent set. Clearly, equivalence can be checked in polynomial time and there are at most~$1+|V(G)|^2\leq 1+N^2$ equivalence classes on any finite set of instances of size at most~$N$ each.

Now, given~$t$ \eqvr-equivalent instances of \IS, say,~$(G_1,k),\ldots, \linebreak[1] (G_t,k)$ where each~$G_i$ has exactly~$n$ vertices, we construct an ILP with a \emph{core part} consisting of~$n^{\Oh(1)}$ variables and constraints plus a set of roughly~$n^2$~$4$-protrusions. (Should the instances be malformed then the correct answer is \no and we may return any constant-size infeasible ILP.) Our construction rules out any form of kernel or compression for the target problem to size polynomial in~$(n+\log t)$ (unless \containment). 

\emph{Construction.} The core part of the \ILPF instance is formed by a straightforward set of constraints that checks whether a set of~$n$ variables encodes an independent set of size at least~$k$, subject to a set of edges that itself is given by variables. For ease of presentation, assume that all input graphs have the same vertex set~$V(G_i)=\{1,\ldots,n\}$ and (different) edge sets~$E(G_i)\subseteq\binom{V}{2}$.
\begin{itemize}
 \item We introduce~$n$ variables~$x_1,\ldots,x_n$ that will encode an independent set; we constrain them by enforcing~$x_i\in\{0,1\}$. Accordingly, since all~$t$ instances seek an independent set of size at least~$k$ we add a constraint
 \begin{align}
  \sum_{i=1}^n x_i\geq k.\label{constraint:crosscomp:totalsize}
 \end{align}
 \item We add~$\binom{n}{2}$ variables~$y_{i,j}$, for all~$1\leq i<j\leq n$ that will contain the adjacency information of a single input graph; we also enforce~$y_{i,j}\in\{0,1\}$. For now we only add constraints that enforce that the independent set encoded in~$x_1,\ldots,x_n$ is consistent with the information in the~$y$-variables. For all~$1\leq i<j\leq n$ we enforce
 \begin{align}
  x_i+x_j+y_{i,j}\leq 2.\label{constraint:crosscomp:edges}
 \end{align}
 \item We add a single variable~$s$ that effectively chooses one input graph. We enforce~$s\in\{1,\ldots,t\}$. This variable is used in the boundary of each protrusion.
\end{itemize}
Now we describe the protrusion part of the instance. For each~$1\leq i<j\leq n$ we add additional variables and constraints whose corresponding graph has treewidth four and such that we retrieve the edge information for graph~$G_s$. In the following, we explain how to build the constraints that encode the edge information for edge~$\{i,j\}$ for all graphs~$G_1,\ldots,G_t$ for some fixed choice of~$i$ and~$j$.
\begin{itemize}
 \item We begin with adding indicator variables~$d^{i,j}_p$ for~$p\in\{1,\ldots,t\}$ with the intention of enforcing
 \begin{align*}
  d^{i,j}_p=\begin{cases}
       0 & \mbox{if~$s=p$,}\\
       1 & \mbox{if~$s\neq p$.}
      \end{cases}
 \end{align*}
 We restrict these variables to domain~$\{0,1\}$.
 Once we have these available it will be straightforward to enforce that~$y_{i,j}$ correctly represents whether~$\{i,j\}\in E(G_s)$, i.e., whether the graph~$G_s$ chosen by~$s\in\{1,\ldots,t\}$ contains the edge~$\{i,j\}$.
 \item As a first step, we add constraints enforcing that~$d_p=1$ for all~$p$ with~$s\neq p$. To this end we add the following constraints for all~$p\in\{1,\ldots,t\}$.
 \begin{align*}
  s&\geq p-t\cdot d^{i,j}_p\\
  s&\leq p+t\cdot d^{i,j}_p
 \end{align*}
 Observe that~$p$ and~$t$ are constants in these constraints. Clearly, if~$d^{i,j}_p=0$ then we must have~$s=p$, so this can occur only for that particular variable~$d^{i,j}_p=d_s$. Unfortunately, this alone does not suffice, since it does not prevent setting \emph{all} variables~$d^{i,j}_p$ to one; we fix this in the next step.
 \item As a second step, we add constraints that are equivalent to enforcing~$\sum_p d^{i,j}_p=t-1$. (Note that we cannot outright add this constraint as its treewidth would be huge.) To this end, we use additional variables~$c^{i,j}_1,\ldots,c^{i,j}_t$ with domain~$\{0,1\}$. The constraints are as follows.
 \begin{align*}
  c^{i,j}_1&=d^{i,j}_1\\
  c^{i,j}_i&=c^{i,j}_{\ell-1}+d^{i,j}_\ell-1\quad\forall \ell\in\{2,\ldots,t\}\\
  c^{i,j}_t&=0
 \end{align*}
 Adding up all equations except for~$c^{i,j}_t=0$ and subtracting~$c^{i,j}_1+\ldots+c^{i,j}_{t-1}$ from both sides yields~$c^{i,j}_t=\sum_p d^{i,j}_p-(t-1)$. Thus, combined with~$c^{i,j}_t=0$, this enforces~$\sum_p d^{i,j}_p=t-1$.
 
 We observe that both sets of constraints together enforce the desired behavior for all variables~$d^{i,j}_p$. Since~$s\in\{1,\ldots,t\}$, we have~$t-1$ choices of~$p$ with~$s\neq p$. For the corresponding variables~$d^{i,j}_p$ we enforced that~$d^{i,j}_p=1$. Due to the constraint~$\sum_p d^{i,j}_p = t-1$ we must therefore have~$d^{i,j}_s=0$.
 \item Now we may use the~$d^{i,j}_p$ variables to enforce that the presence of edge~$\{i,j\}$ in graph~$G_s$ is correctly stored in~$y_{i,j}$. It suffices to add the following constraints (i.e., one constraint is added for each~$p\in\{1,\ldots t\}$ with the choice depending on whether~$\{i,j\}\in E(G_p)$).
 \begin{align*}
 \begin{cases}
  y_{i,j}\geq 1-d^{i,j}_p & \mbox{if~$\{i,j\}\in E(G_p)$,}\\
  y_{i,j}\leq 0+d^{i,j}_p & \mbox{if~$\{i,j\}\notin E(G_p)$.}
 \end{cases}
 \end{align*}
 Recall that~$y_{i,j}$ has domain~$\{0,1\}$ and thus whenever~$d^{i,j}_p=1$ there is no additional restriction for~$y_{i,j}$. In the (unique) case that~$d^{i,j}_p=0$ (together with~$y_{i,j}\in\{0,1\}$) this clearly enforces
 \begin{align*}
  y_{i,j}=\begin{cases}
           1 & \mbox{if~$\{i,j\}\in E(G_p)$,}\\
           0 & \mbox{if~$\{i,j\}\notin E(G_p)$.}
          \end{cases}
 \end{align*}
 Since we ensured that~$d^{i,j}_p=0$ if and only if~$s=p$ this correctly enforces that~$y_{i,j}$ carries the information about presence of~$\{i,j\}$ in~$G_s$, as desired.
 \item Finally, let us check that the added variables together with~$y_{i,j}$ and~$s$ indeed correspond to a protrusion in the corresponding Gaifman graph. Clearly, only~$s$ and~$y_{i,j}$ occur in any further constraints so the boundary size is two. It can be checked that the treewidth of the subgraph induced by vertices of the present variables is at most four. (Key: Make a path decomposition of bags~$\{s,y_{i,j},c^{i,j}_{p-1},c^{i,j}_p,d^{i,j}_p\}$ for increasing~$p$.)
\end{itemize}
This completes our construction. Since we claim that the \ILPF problem does not even admit a polynomial kernelization or compression when a protrusion decomposition is given along with the input, we have to construct a protrusion decomposition as part of the cross-composition. It is defined as follows. The set~$Y_0$ has size~$n + \binom{n}{2}+1$ and consists of the variables~$x_i$ for~$i \in [n]$, variables~$y_{i,j}$ for~$1 \leq i < j \leq n$, and the variable~$s$. For each choice of~$1 \leq i < j < n$ we create a set~$Y_{\ldots}$ in the protrusion decomposition that contains the variables for the protrusion involving~$s$ and~$y_{i,j}$ that was described above. The neighborhood of each such set~$Y_{\ldots}$ in~$Y_0$ consists of~$s$ together with~$y_{i,j}$ and therefore has size at most two. As argued above, the treewidth of the graph induced by the protrusion and its neighborhood is at most four. We need~$\ell = \binom{n}{2}$ different sets in the protrusion decomposition. It is therefore a~$(k',5)$-protrusion decomposition for~$k' = \max(|Y_0|, \ell) \in \Oh(n^2)$. To bound the number~$m'$ of constraints that affect only the variables of~$Y_0$, observe that the only constraints whose support is a subset of~$Y_0$ are the single constraint~\cref{constraint:crosscomp:totalsize}, the~$\binom{n}{2}$ constraints~\cref{constraint:crosscomp:edges}, and the~$\Oh(|Y_0|)$ domain-enforcing constraints on~$Y_0$. We therefore have~$m' \in \Oh(n^2)$. Since the only variable that does not have a binary domain is~$s$ with domain~$\{1, \ldots, t\}$, we have~$d := t$. For the total parameter value of the constructed instance we therefore have~$k' + m' + d \in \Oh(n^2 + \log t)$, which is polynomial in the size of the largest input instance plus~$\log t$. It is therefore suitably bounded for a cross-composition. To complete the lower bound it suffices to prove that the constructed instance acts as the logical OR of the inputs.

\emph{Correctness.} We will keep this brief since we already discussed the workings of the created instance. If at least one graph~$G_{s^*}$ has an independent set of size at least~$k$, then set the indicators~$x_i$ accordingly and set~$s=s^*$. Furthermore, set the edge indicators~$y_{i,j}$ according to~$E(G_{s^*})$. It remains to verify that there are feasible values for~$c^{i,j}_p,d^{i,j}_p$ for all~$1 \leq i < j \leq n$ and~$1 \leq p \leq t$. It can be verified that the following values are feasible (taking into account~$s=s^*$ and the value of~$y_{i,j}$ variables).
\begin{align*}
c^{i,j}_p=\begin{cases}
    1 & \mbox{if~$p<s$,}\\
    0 & \mbox{if~$p\geq s$.}
    \end{cases}&&d^{i,j}_p=\begin{cases}
    1 & \mbox{if~$p\neq s$,}\\
    0 & \mbox{if~$p=s$.}
    \end{cases}
\end{align*}
Conversely, we already discussed that for any choice of~$s$ the variables~$y_{i,j}$ correspond to the presence of edge~$\{i,j\}$ in~$G_s$. Thus, the remaining constraints correctly verify the presence of an independent set of size at least~$k$ in~$G_s$. This implies that~$(G_s,k)$ is \yes. This completes the proof of Theorem~\ref{theorem:lb:largedomainvar}.
\end{proof}

Intuitively, the parameterization chosen in the theorem implies that everything can be bounded to size polynomial in the parameters except for the variables in~$Y_1,\ldots,Y_\ell$ and the (encoding size of the) constraints that are fully contained in protrusions (recall that constraints give cliques in~$G(A)$).
To put this lower bound into context, we prove that a more general (and less technical) variant is fixed-parameter tractable.

\begin{theorem}\label{theorem:ilpfvariant:fewunbounded:fpt}
The following variant of \ILPF is \FPT: We allow a constant number~$c$ of variables with polynomially bounded domain; all other variables have domain size~$d$. Furthermore, there is a specified set of variables~$S\subseteq\{x_1,\ldots,x_n\}$ such that the graph~$G(A)-S$ has bounded treewidth. The parameter is~$d+|S|$.
\end{theorem}

\begin{proof}
This follows readily from Theorem~\ref{theorem:boundeddomain:boundedtreewidth:dp} once we take care of the at most~$c$ high-domain variables. To do so, we can simply branch over all possible assignments to the variables since the total number of choices of~$c$ values from a polynomially bounded domain is itself polynomially bounded. Thus, for each choice we replace the high-domain variables by the chosen values and (after rearranging) obtain an ILP on slightly fewer variables that all have bounded domain. Now, let us note that the treewidth of~$G(A)$ is initially bounded by~$|S|+\Oh(1)$ and that the Gaifman graph for the new ILP has at most the same treewidth (effectively it is obtained by deleting the vertices corresponding to high-domain variables). Thus, we may run the algorithm from Theorem~\ref{theorem:boundeddomain:boundedtreewidth:dp} with~$w\leq|S|+\Oh(1)$ to obtain the claimed result.
\end{proof}

\section{Totally unimodular subproblems}\label{section:tum}

Recall that a matrix~$A$ is \emph{totally unimodular} (TU) if and only if each square submatrix of~$A$ has determinant in~$\{-1,0,1\}$; this requires that~$A\in\{-1,0,1\}^{m\times n}$ since any single entry defines a one by one square submatrix. If an ILP is given by~$Ax\leq b$ where~$A$ is totally unimodular and~$b$ is integral, then all extremal points of the corresponding polyhedron are integral. Thus, solving the relaxed~$LP$ suffices for feasibility and even for optimizing any function~$c^Tx$ subject to~$Ax\leq b$.

We say that a matrix~$A$ is \emph{totally unimodular plus~$p$ entries} if~$A$ can be obtained from a totally unimodular matrix by replacing any~$p$ entries by new values. (This is more restrictive than the equally natural definition of adding~$p$ arbitrary rows or columns.) We note that \ILPF is \FPT with respect to parameter $p+d$, where $d$ bounds the domain:
It suffices to try all $d^p$ assignments for variables whose column in $A$ has at least one modified entry. After simplification the obtained system is TU and, thus, existence of a feasible assignment for the remaining entries can be tested in polynomial time.

Our following result shows that, despite fixed-parameter tractability for parameter~$p + d$, the existence of a polynomial kernelization is unlikely; this holds already when~$d=2$.

\begin{theorem}\label{theorem:tulb}
\problem{ILP Feasibility} restricted to instances~$(A,b,p)$ where~$A$ is totally unimodular plus~$p$ entries does not admit a kernel or compression to size polynomial in~$p$ unless \containment, even if all domains are~$\{0,1\}$.
\end{theorem}

\begin{proof}
We reduce from the \HSN problem, in which we are given a set~$U$, a set~$\F\subseteq 2^U$ of subsets of~$U$, and an integer~$k\in\N$, and we have to decide whether there is a choice of at most~$k$ elements of~$U$ that intersects all sets in~$\F$; the parameter is~$n:=|U|$. Dom et al.~\cite{DomLS14} proved that \HSN admits no polynomial kernelization or compression (in terms of~$n$) unless \containment. We present a polynomial-parameter transformation from \HSN to \ILPF{}($p$) with domain size~$2$. The ILP produced by the reduction will be of the form~$Ax\leq b$ where~$A$ is totally unimodular plus~$p=n$ entries. Thus, any polynomial kernelization or compression in terms of~$p$ would give a polynomial compression for \HSN and, thus, imply \containment as claimed.

\emph{Construction.} Let an instance~$(U,\F,k)$ be given. We construct an ILP with~$0/1$-variables that is feasible if and only if~$(U,\F,k)$ is \yes for \HSN.
\begin{itemize}
 \item Our ILP has two types of variables:~$x_{u,F}$ for all~$u\in U,F\in\F$ and~$x_u$ for all~$u\in U$. For all variables we enforce domain~$\{0,1\}$ by~$x_{u,F}\geq 0$,~$x_{u,F}\leq 1$,~$x_u\geq 0$, and~$x_u\leq 1$ for~$u\in U$ and~$F\in\F$.
 \item The variables~$x_{u,F}$ are intended to encode what elements of~$u$ ``hit'' which sets~$F\in\F$. We enforce that each set~$F\in\F$ is ``hit'' by adding the following constraint.
 \begin{align}
  1\leq \sum_{u\in F}x_{u,F}&&\forall F\in\F  \label{constraint:checkhit}
 \end{align}
 \item The variables~$x_u$ are used to control which variables~$x_{u,F}$ may be assigned values greater than zero; effectively, they correspond to the choice of a hitting set from~$U$. Control of the~$x_{u,F}$ variables comes from the following constraints.
 \begin{align}
	\sum_{F\in\F}x_{u,F}\leq |\F|\cdot x_u && \forall u\in U \label{constraint:controlxufbyxu}
 \end{align}
 Additionally, we constrain the sum over all~$x_u$ to be at most~$k$, in line with the concept of having~$x_u$ select a hitting set of size at most~$k$.
 \begin{align}
  \sum_{u\in U} x_u\leq k \label{constraint:sumxu}
 \end{align}
\end{itemize}
Clearly, the construction can be performed in polynomial time. Let us now prove correctness.

\emph{Correctness.} Assume that~$(U,\F,k)$ is \yes for \HSN and let~$S\subseteq U$ be a hitting set for~$\F$ of size at most~$k$. Set all~$x_u$ to~$1$ if~$u\in S$ and to~$0$ otherwise; this fulfills \cref{constraint:sumxu}. For each~$F\in\F$ there is at least one~$u\in S\cap F$ as~$S$ is a hitting set for~$\F$ and we set the variable~$x_{u,F}$ to~$1$; all other~$x_{u,F}$ are set to~$0$. Thus, we satisfy all constraints \cref{constraint:checkhit}. Clearly, only~$x_{u,F}$ with~$u\in S$ (and hence~$x_u=1$) receive value~$1$ and, thus, this assignment fulfills also \cref{constraint:controlxufbyxu}. Finally, we note that all variables receive values from~$\{0,1\}$, as required. Thus, our assignment is a feasible solution for the ILP, as claimed.

Conversely, assume that the ILP is feasible and fix any feasible assignment to all variables. Define a set~$S\subseteq U$ by picking those~$u\in U$ with~$x_u= 1$. Clearly, by \cref{constraint:sumxu} and domain~$\{0,1\}$, there are at most~$k$ such variables and, hence,~$|S|\leq k$. We will prove that~$S$ is a hitting set for~$\F$, so fix any set~$F\in\F$. By constraint \cref{constraint:checkhit}, at least one variable $x_{u,F}$ with $u\in F$ must take value~$1$. The corresponding element~$u$ must be in~$S$ since constraint~$\cref{constraint:controlxufbyxu}$ plus non-negativity can only be fulfilled when~$x_u\geq 1$ (as the left-hand side sum takes value greater than~$0$). By choice of~$S$ this implies~$u\in S$. Overall we get that~$S$ contains some element~$u\in S\cap F$ for all sets~$F\in\F$, implying that~$S$ is indeed a hitting set for~$\F$.

In the remainder of the proof, we show that the constraints can be written as~$A\x\leq b$ where~$A$ is totally unimodular plus $n$ entries (here~$\x$ stands for the vector of all variables~$x_{u,F}$ and~$x_u$ over all~$u\in U$ and~$F\in\F$). First, we need to write our constraints in the form
$A\begin{pmatrix}x_{U,\F}\\x_U\end{pmatrix}\leq b$,
where, e.g.,~$x_U$ stands for the column vector of all variables~$x_u$ with~$u\in U$. For now, we translate constraints~\cref{constraint:checkhit}, \cref{constraint:controlxufbyxu}, and \cref{constraint:sumxu} into this form; domain-enforcing constraints will be discussed later. We obtain the following, where~$(\mathbf{1/0})$ and~$(\mathbf{-1/0})$ are shorthand for submatrices that are entirely~$0$ except for exactly one~$1$ or one~$-1$ per column, respectively.

{
\small
    \vspace{0.1cm}
   \newcommand{\spc}[1]{\makebox[1.2cm]{#1}}
   \begin{align*}
   A'\x=
   \begin{pmatrix}
   \begin{pmatrix} & &  \\ & \spc{$\mathbf{-1/0}$} & \\ & & \end{pmatrix}
   &\quad & 
   \begin{pmatrix}  & \cdots &  \\ \spc{$\vdots$} & \ddots & \vdots \\ \spc{$0$} & \spc{$\cdots$} & \spc{$0$}\end{pmatrix}
   \\
   \begin{pmatrix} & & \\ & \spc{$\mathbf{1/0}$} & \\ \\ \end{pmatrix} & & 
   \begin{pmatrix} \spc{$-|\F|$} & 0 & 0 \\ 0 & \spc{$-|\F|$} & 0 \\ 0 & 0 & \spc{$-|\F|$} \end{pmatrix} \\
   0\cdots0 && \begin{matrix} \spc{1} & & \spc{$\cdots$} & & \spc{1} \end{matrix}
   \end{pmatrix}
   \begin{pmatrix}x_{U,\F}\\x_U\end{pmatrix}
   \leq
   \begin{pmatrix}
   \begin{matrix} -1 \\ \vdots \\ -1 \end{matrix} \\
   \begin{matrix} 0 \\ \vdots \\ 0 \end{matrix} \\
   k
   \end{pmatrix}
   &&
   \begin{array}{l}
   \phantom{-1} \\ \cref{constraint:checkhit} \phantom{\vdots} \\ \phantom{-1} \\ \phantom{0} \\ \cref{constraint:controlxufbyxu} \phantom{\vdots} \\ \phantom{0} \\ \cref{constraint:sumxu} 
   \end{array}
   \end{align*}
}

In this expression, the columns of the matrix~$A'$ is are split into two groups. The first group contains the~$|\F| \times |U|$ columns for the variables~$x_{u,F}$ for~$u \in U$ and~$F \in \F$, while the second group contains the~$|U|$ columns for the variables~$x_u$ with~$u \in U$. Let us denote by~$\hat{A'}$ the matrix obtained from~$A'$ by replacing all~$n$ entries~$-|\F|$ by zero.

{\small
\begin{align*}
\hat{A'}=
\begin{pmatrix}
\mathbf{(-1/0)} &\quad & \mathbf{0}\\
\mathbf{(1/0)} && \mathbf{0}\\
0\cdots0 && 1\cdots1
\end{pmatrix}
\end{align*}

}

It is known that any matrix over~$\{-1,0,1\}$ in which every column has at most one entry~$1$ and at most one entry~$-1$, is totally unimodular (cf.~\cite[Theorem 13.9]{SchrijverBook}). Since~$\hat{A'}$ is of this form, it is clear that~$\hat{A'}$ is totally unimodular. To obtain the whole constraint matrix~$A$ we need to add rows corresponding to domain-enforcing constraints for all variables, and reset the~$-|\F|$ values that we replaced by zero. Clearly, putting back the latter breaks total unimodularity (and this is why~$A$ is only almost TU), but let us add everything else and see that the obtained matrix~$\hat{A}$ is totally unimodular. The domain-enforcing constraints affect only one variable each and, thus, each of them corresponds to a row in~$\hat{A}$ that contains only a single nonzero entry of value~$1$ or~$-1$. It is well known that adding such rows (or columns) preserves total unimodularity. (The determinant of any square submatrix containing such a row can be reduced to that of a smaller submatrix by expanding along a row that has only one nonzero of~$1$ or~$-1$.)

Finally, $\hat{A}$ and $A$ are only distinguished by the~$n$ entries of value~$-|\F|$ that are present in~$A$ but which are~$0$ in~$\hat{A}$. Since~$\hat{A}$ is totally unimodular, it follows that~$A$ is totally unimodular plus~$p=n$ entries, as claimed.
\end{proof}

Complementing Theorem~\ref{theorem:tulb}, 
we prove that TU subsystems of an ILP can be reduced to a size that is polynomial in the domain, with degree depending on the number of variables that occur also in the rest of the ILP. We again phrase this in terms of replacing boundaried ILPs and prove that any $r$-boundaried TU subsystem can be replaced by a small equivalent system of size polynomial in the domain $d$ with degree depending on $r$.

\begin{lemma}\label{lemma:tum:boundariedilp:replacement}
There is an algorithm with the following specifications: (1) It gets as input an~$r$-boundaried ILP~$(A,b,x_{t_1}, \ldots, x_{t_r})$ with domain size~$d$, with~$A\in\Z^{m\times n}$,~$b\in\Z^m$, and such that the restriction of~$A$ to columns~$[m]\setminus\{t_1,\ldots,t_r\}$ is totally unimodular. (2) Given such an input it takes time $\Oh(d^r \cdot g(n,m) + d^{2r})$ where~$g(n,m)$ is the runtime for an LP solver for determining feasibility of a linear program with~$n$ variables and~$m$ constraints. (3) Its output is an equivalent~$r$-boundaried ILP~$(A',b',x'_{t'_1},\ldots,x'_{t'_r})$ of domain size~$d$ containing~$\Oh(r \cdot d^{r})$ variables and constraints, and all entries of~$(A',b')$ in~$\{-d,\ldots,d\}$.
\end{lemma}

\begin{proof}
Similar to the treewidth case (Lemma~\ref{lemma:tw:boundariedilp:replacement}) this is a two-step process where we first determine a set~$\L$ of infeasible partial assignments~$(a_1,\ldots,a_r)$, i.e., such that no feasible assignment for~$Ax\leq b$ has~$x_{t_1}=a_1,\ldots,x_{t_r}=a_r$, and then encoding those directly using new constraints and variables. The second part can be done just as for Lemma~\ref{lemma:tw:boundariedilp:replacement} but we use a different routine for finding infeasible assignments since we no longer have bounded treewidth of~$G(A)$.

Concretely, for every~$(a_1,\ldots,a_r)\in\{0, \ldots, d-1\}^r$ we set~$x_{t_1}=a_1,\ldots,x_{t_r}=a_r$ and simplify the system to~$A'x'\leq b'$ where~$x'$ is the vector of all remaining variables. Here~$A'$ is simply~$A$ restricted to columns~$[n]\setminus\{t_1,\ldots,t_r\}$ and
\begin{align*}
b'[i]=b[i]-\sum_{j\in\{t_1,\ldots,t_r\}} A[i,j]a_j, && \mbox{for~$i\in[m]$.}
\end{align*}
We recall that~$A'$ is totally unimodular and note that~$b'$ is integer, implying that we can test feasibility of~$A'x'\leq b'$ in polynomial-time using an LP-solver (e.g., the ellipsoid method).

Again, this gives us a list of all infeasible assignments to~$x_{t_1},\ldots,x_{t_r}$ and we proceed as in the proof of Lemma~\ref{lemma:tw:boundariedilp:replacement} to output an equivalent ILP with~$r$ boundary variables. This completes the proof.
\end{proof}

Lemma~\ref{lemma:tum:boundariedilp:replacement} implies that if $Y_0$ is a set of (at most) $p$ variables whose removal makes the remaining system TU, then the number of variables in the system can efficiently be reduced to a polynomial in~$d+p$ with degree depending on $r$, if each TU subsystem depends on at most $r$ variables in $Y_0$. To get this, it suffices to apply Lemma~\ref{lemma:tum:boundariedilp:replacement} once for each choice of at most $r$ boundary variables in $Y_0$. (Note that without assuming a bounded value of $r$ we only know $r\leq p$, so the worst-case bound obtained is not polynomial, but exponential, in $p+d$.)


\section{Observations on TU subsystems with small boundary}\label{section:tudiscussion}

In this section, we discuss some of the differences between TU subsystems and those of bounded treewidth (both with small boundary); concretely, we argue that the former are somewhat easier to handle than the latter. To this end, we discuss how one can improve Lemma~\ref{lemma:tum:boundariedilp:replacement} in comparison to its bounded-treewidth counterpart (Lemma~\ref{lemma:tw:boundariedilp:replacement}). For intuition, consider the case of having a single boundary variable only, say $x_1$. Let $0\leq i<j\leq d$ such that there are integer feasible assignments for the subsystem with $x_1=i$ and $x_1=j$, respectively. By convexity, this gives fractionally feasible assignments for all $x_1=i+\delta(j-i)$, for $0\leq\delta\leq 1$. For every choice of $\delta$ such that $x_1$ is integer, we get a fractionally feasible assignment for the subsystem. Since the system with value of $x_1$ plugged in simplifies to a TU system on the remaining variables, it follows that there are also integer feasible assignments that are consistent with this value of $x_1$. Thus, the feasible values for $x_1$ are equal to $\{i,\ldots,j\}$ for some $0\leq i<j\leq d$, which takes only $\Oh(\log d)$ bits to encode (by two simple constraints). In the same way, we can rule out a lower bound as for the treewidth case (Theorem~\ref{theorem:lb:largedomainvar}).
There we used two boundary variables with domains $\{0,1\}$ and $\{1,\ldots,t\}$, and relied heavily on the fact that the feasible boundary assignments were not simply the integer points in a convex polytope; thus, we effectively encoded $t$ bits of information. For the TU case, fixing the first variable yields a consecutive feasible interval for the second variable; the two intervals can be encoded in $\Oh(\log t)$ bits.

In general, total unimodularity of the subsystem (minus the boundary) implies that we are only interested in the boundary assignments that have fractionally feasible completions in the subsystem. The argument of the previous paragraph generalizes to larger boundary sizes: if~$\mathbf{x}$ and~$\mathbf{x'}$ are feasible vectors of boundary assignments, then any convex combination $\mathbf{x} + \delta (\mathbf{x'} - \mathbf{x})$ for~$0 \leq \delta \leq 1$ that is an integer vector can be extended to a feasible assignment. Let~$\mathcal{B}$ denote the set that contains, for every feasible integer assignment, its projection to the boundary variables. Then~$\mathcal{B}$ is closed under taking convex combinations that result in integer vectors. It follows that an integer boundary assignment is feasible if and only if it lies in the convex hull of~$\mathcal{B}$. Thus, the (encoding) complexity of feasible assignments to the boundary is governed by known upper and lower bounds for the integer hull of polytopes. According to a survey of B\'ar\'any~\cite[Theorem 7.1]{Barany_survey} the maximum number of vertices for an integer polyhedron in~$r$ dimensions whose points have coordinates from~$\{0, \ldots, d-1\}$ is $\Theta(d^{r\frac{r-1}{r+1}})$, which converges to $\Theta(d^{r-2})$ (from above) for growing $r$. This implies that the feasible boundary assignments to an $r$-boundaried ILP of domain~$d$, whose non-boundary variables induce a TU system, can be described in~$\Theta(d^{r\frac{r-1}{r+1}} r \log d)$ bits by writing down the coordinates of all the vertices. The same polyhedral bound can be used to show that~$\Omega(d^{r\frac{r-1}{r+1}})$ bits are necessary: An integral polyhedron in~$[0\ldots d-1]^r$ with~$N$ vertices forms the solution space of an ILP with~$r$ variables. Any single vertex of the polyhedron can be cut off from the feasible region by an additional constraint: Take a halfspace that intersects the polyhedron only at that vertex and move it slightly toward the interior. The resulting integer polyhedron is again described by an ILP. As there is a different ILP for each of the~$2^N$ subsets of vertices that can be cut off, it follows that at least~$N$ bits are needed to describe such an ILP. Since there are integral polyhedra with~$r$ variables, domain $\{0,\ldots,d-1\}$ and~$N \in \Omega(d^{r\frac{r-1}{r+1}})$ vertices, it follows that this number lower bounds the number of bits needed to encode an $r$-boundaried ILP whose remaining variables form a TU system. (Note that this lower bound applies even if the TU part is empty.) It is interesting to note the contrast with $r$-boundaried ILP's of treewidth~$\Theta(r)$, for which~$\Theta(d^r)$ bits are necessary by Proposition~\ref{proposition:tw:protrusion:lowerbound} and trivially sufficient.


\section{Discussion and future work}\label{section:discussion}

We have studied the effect that subsystems with bounded treewidth or total unimodularity have regarding kernelization of the \textsc{ILP Feasibility} problem. We show that if such subsystems have a constant-size boundary to the rest of the system, then they can be replaced by an equivalent subsystem of size polynomial in the domain size (with degree depending on the boundary size). Thus, if an \ILPF instance can be decomposed by specifying a set of $p$ shared variables whose deletion (or replacing with concrete values) creates subsystems that are all TU or bounded treewidth and have bounded dependence on the $p$ variables, then this can be replaced by an equivalent system whose number of variables is polynomial in $p$ and the domain size $d$. We point out that for the case of binary variables (at least in the boundary) the replacement structures get much simpler, using no additional variables and with only a single constraint per forbidden assignment. Using a similar approach and binary encoding for boundary variables should reduce the number of additional variables to $\Oh(\log d)$ per boundary variable.

Complementing this, we established several lower bounds regarding limitations of replacing such subsystems. Inherently, the replacement rules rely on having subsystems with small boundary size for giving polynomial bounds. We showed that this is indeed necessary by giving lower bounds for fairly restricted settings where we do not have the guarantee of constant boundary size, independent of the means of data reduction. In the case of treewidth we could also show that boundaries with only one large-domain variable can be a provable obstacle. For the case of totally unimodular subsystems the discussion in the previous section shows that these behave in a slightly simpler way than bounded-treewidth subsystems: By ad hoc arguments we can save a factor of $d$ in the encoding size by essentially dropping the contribution of any one boundary variable; thus we rule out a lower bound proof for the case of one boundary variable having large domain. Asymptotically, we can save a factor of almost $d^2$, which is tight. It would be interesting whether having two or more large domain variables (in a boundary of constant size) would again allow a lower bound against kernelization.

A natural extension of our work is to consider the optimization setting where we have to minimize or maximize a linear function over the variables, and may or may not already know that the system is feasible. In part, our techniques are already consistent with this since the reduction routine based on treewidth dynamic programming or optimization over a TU subsystem can be easily augmented to also optimize a target function over the variables. A technical caveat, however, is the following: If we simplify a protrusion, then along with each feasible assignment to the boundary, we have to store the best target function contribution that could be obtained with the variables that are removed; this value can, theoretically, be unbounded in all other parameters. If a binary encoding of such values is sufficiently small (or if the needed space is allowed through an additional parameter), then our results also carry over to optimization. Apart from that, a rigorous analysis of both weight reduction techniques and possible lower bounds is left as future work.

\bibliographystyle{abbrvurl}
\bibliography{ilpkernel_arxiv}

\end{document}